\documentclass[a4paper,UKenglish,cleveref, autoref, thm-restate]{lipics-v2021}
\usepackage{booktabs}
\usepackage{xspace}
\usepackage{comment}
\usepackage{thmtools}
\usepackage{booktabs}
\nolinenumbers
\hideLIPIcs  %uncomment to remove references to LIPIcs series (logo, DOI, ..), e.g. when preparing a pre-final version to be uploaded to arXiv or another public repository

\usepackage{microtype} % if unwanted, comment out

\newcommand*\rot[2]{\multirow{#1}{*}{\rotatebox[origin=c]{90}{\parbox[c]{1.2cm}{\centering #2}}}}
\newcommand*\rotlong[2]{\multirow{#1}{*}{\rotatebox[origin=c]{90}{\parbox[c]{4cm}{\centering #2}}}}

\newcommand {\R} {\mathbb{R}\xspace}

\newcommand{\Exp}{\mathbb{E}\xspace}

\newtheorem*{conjecture*}{Conjecture}

\newcommand{\mkmcal}[1]{\ensuremath{\mathcal{#1}}\xspace}
\newcommand{\C}{\mkmcal{C}}
\newcommand{\A}{\mkmcal{A}}
\newcommand{\I}{\mkmcal{I}}
\newcommand{\D}{\mkmcal{D}}

\newcommand{\X}{\ensuremath{\chi}}

\newcommand{\B}{\mkmcal{B}}

\newcommand{\T}{\mkmcal{T}}
\newcommand{\Tr}{\ensuremath{\mathfrak{T}}\xspace}

\newcommand{\RR}{\mkmcal{R}}
\renewcommand{\S}{\mkmcal{S}}

\renewcommand{\L}{\mkmcal{L}}

\newcommand{\E}{\mkmcal{E}}
\newcommand{\NN}{\mkmcal{N}}
\newcommand{\VD}{\mkmcal{VD}}
\newcommand{\sizeA}{\ensuremath{|\A |}}
\newcommand{\sizeB}{\ensuremath{|\B |}}
\newcommand{\upto}[2]{\ensuremath{\{#1,\dots,#2\}}}

\newcommand{\Rel}{\ensuremath{\mathit{Rel}}\xspace}

\newcommand{\querytime}{\ensuremath{Q}\xspace}
\newcommand{\spacecomplexity}{\ensuremath{S}\xspace}

\newcommand{\etal}{et al.\xspace}
\newcommand{\Matousek}{Matou{\v s}ek\xspace}

\newcommand{\eps}{\ensuremath{\varepsilon}\xspace}
\newcommand{\SPM}{\ensuremath{\mathit{SPM}}\xspace}

\newcommand{\Xis}{\ensuremath{\Xi}\xspace}
\newcommand{\Xit}{\ensuremath{\Xi_\Delta}\xspace}
\newcommand{\storage}{\ensuremath{S}}
\newcommand{\query}{\ensuremath{Q}}

\newcommand{\geod}{\ensuremath{\pi}\xspace}

\newcommand{\TwoPointShortestPath}{\textsc{Two-Point-Shortest-Path}\xspace}

\newcommand{\LowerEnvelope}{\textsc{Lower Envelope}\xspace}

\title{Towards Space Efficient Two-Point Shortest Path Queries in a Polygonal Domain}
\author{Sarita de Berg}{Department of Information and Computing Sciences, Utrecht University, The Netherlands}{s.deberg@uu.nl}{}{}%TODO mandatory, please use full name; only 1 author per \author macro; first two parameters are mandatory, other parameters can be empty. Please provide at least the name of the affiliation and the country. The full address is optional

\author{Tillmann Miltzow}{Department of Information and Computing Sciences, Utrecht University, The Netherlands}{t.miltzow@uu.nl}{}{is generously supported by the Netherlands Organisation for Scientific Research (NWO) under project  no. VI.Vidi.213.150.}

\author{Frank Staals}{Department of Information and Computing Sciences, Utrecht University, The Netherlands}{f.staals@uu.nl}{}{}

\authorrunning{S. de Berg, T. Miltzow, and F. Staals} %TODO mandatory. First: Use abbreviated first/middle names. Second (only in severe cases): Use first author plus 'et al.'

 %Please choose ACM 2012 classifications from https://dl.acm.org/ccs/ccs_flat.cfm

  %TODO mandatory; please add comma-separated list of keywords

\Copyright{Sarita de Berg, Tillman Miltzow, and Frank Staals} %TODO mandatory, please use full first names. LIPIcs license is "CC-BY";  http://creativecommons.org/licenses/by/3.0/

\authorrunning{S. de Berg, T. Miltzow, and F. Staals}

\begin{document}

\maketitle

\begin{abstract}
We devise a data structure that can answer shortest path queries for two query points in a polygonal domain $P$ on $n$ vertices.  For any $\eps > 0$, the space complexity of the data structure is $O(n^{10+\eps})$ and queries can be answered in $O(\log n)$ time. Alternatively, we can achieve a space complexity of $O(n^{9+\eps})$ by relaxing the query time to $O(\log^2 n)$. This is the first improvement upon a conference paper by Chiang and Mitchell~\cite{chiang99two_point_euclid_short_path_queries_plane} from 1999.  They present a data structure with $O(n^{11})$ space complexity and $O(\log n)$ query time. Our main result can  be extended to include a space-time trade-off. Specifically, we devise data structures with $O(n^{9+\varepsilon}/\hspace{1pt} \ell^{4 + O(\eps)})$ space complexity and $O(\ell \log^2 n )$ query time, for any integer $1 \leq \ell \leq n$.

Furthermore, we present improved data structures for the special case where we restrict one (or both) of the query points to lie on the boundary of $P$. When one of the query points is restricted to lie on the boundary, and the other query point is unrestricted, the space complexity becomes $O(n^{6+\varepsilon})$ and the query time $O(\log^2n)$. When both query points are on the boundary, the space complexity is decreased further to $O(n^{4+\eps})$ and the query time to $O(\log n)$, thereby improving an earlier result of Bae and Okamoto.
\end{abstract}

\vspace{1cm}

\begin{figure}[tbp]
    \centering
        \includegraphics{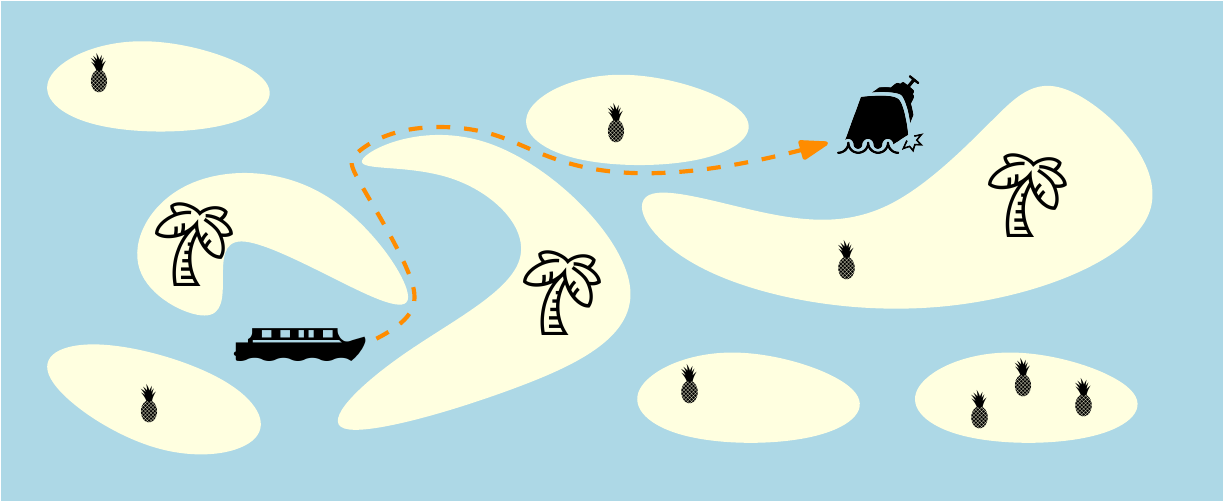}
    \caption{A tangible example of a two-point shortest path problem: finding the shortest path among islands for a boat to an emergency.}
    \label{fig:boat}
\end{figure}

\newpage
\setcounter{page}{1}
\section{Introduction}
% We define the \TwoPointShortestPath problem as the 
% data structure problem with the input being 
% a polygonal domain $\P$ on $n$ vertices and we permit queries of 
% two points $s,t$ and we expect 
% to receive the shortest path within $\P$ from $s$ to $t$.

In the \TwoPointShortestPath problem, we are given a polygonal
domain $P$ with $n$ vertices, possibly containing holes, and we wish to store $P$ so that given
two query points $s,t \in P$ we can compute their \emph{geodesic
  distance} $d(s,t)$, i.e.\ the length of a shortest path fully
contained in $P$, efficiently. After obtaining this distance, the
shortest path can generally be returned in $O(k)$ additional time,
where $k$ denotes the number of edges in the path. We therefore focus
on efficiently querying the distance $d(s,t)$. 

\begin{figure}[bt]
  \centering
  \includegraphics{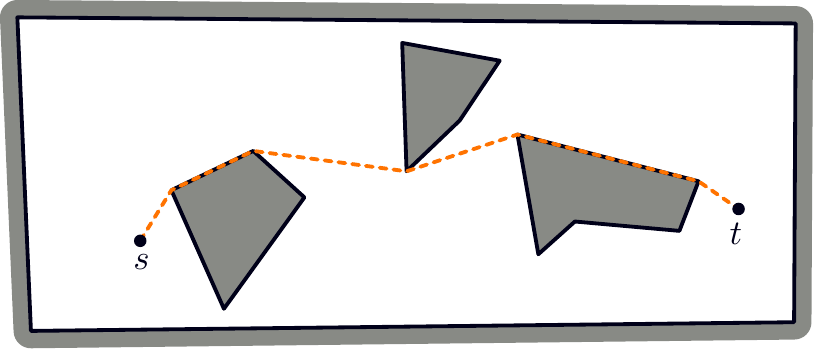}
  \caption{Given $P$ and the query points $s,t$ we want to compute the shortest path efficiently.}
  \label{fig:problem}
\end{figure}

\subparagraph{Tangible example.}
As an example of the relevance of the problem, consider a boat in the
sea surrounded by a number of islands, see Figure~\ref{fig:boat}. Finding the fastest route to an
emergency, such as a sinking boat, corresponds to finding the shortest
path among obstacles, i.e.\ in a polygonal domain.
This is just one of many examples where finding the shortest path in a polygonal domain
is a natural model of a real-life situation, which makes it an interesting problem to study.

\subparagraph{Motivation.}  The main motivation to study the \TwoPointShortestPath problem is that it is a very natural problem. It is
central in computational geometry, and forms a basis for many other
problems. The problem was solved optimally for simple polygons
(polygonal domains without holes) by Guibas and
Hershberger~\cite{ShortestPathSimplePolygon}, and turned out to be a
key ingredient to solve many other problems in simple polygons. A few
noteworthy examples are data structures for geodesic Voronoi
diagrams~\cite{oh19optim_algor_geodes_neares_voron}, farthest-point
Voronoi diagrams~\cite{wang21optim_deter_algor_geodes_farth},
$k$-nearest neighbor
searching~\cite{agarwal18improv_dynam_geodes_neares_neigh,berg21dynam_data_struc_neares_neigh_queries},
and more~\cite{eades20trajec_visib,KineticVoronoi}. In a
  polygonal domain, a two-point shortest path data structure is also
  the key subroutine in computing the geodesic
  diameter~\cite{Bae13_geodesic_diameter} or the geodesic
  center~\cite{Wang18_geodesic_center}.

%%%%%%%%%%%%%%%%%%%%%%%%%%%%%%%%%%%%%
\subsection{Related work}
\label{sub:Related-Work}
%%%%%%%%%%%%%%%%%%%%%%%%%%%%%%%%%%%%%
Chiang and
Mitchell~\cite{chiang99two_point_euclid_short_path_queries_plane}
announced a data structure for the \TwoPointShortestPath problem in
polygonal domains at SODA 1999. They use $O(n^{11})$ space and achieve
a query time of $O(\log n)$. They also present another data structure
that uses ``only'' $O(n^{10}\log n)$ space, but $O(\log^2n)$ query
time. However, the paper refers to the full version for some of the details of these data structures, which never appeared.
%\till{Maybe something like this:We owe a debt of gratitude to [Author(s), 1999] for their pioneering work in this area. However, as we aim to extend and solidify these findings, it's worth mentioning that the original paper was not published in a peer-reviewed journal, and a comprehensive version is not available for further scrutiny.}
Since then, there have been no improvements on the two-point
shortest path problem in its general form. Instead, related and
restricted versions were considered. We briefly discuss the most
relevant ones. Table~\ref{tab:related-work} gives an overview of the related results.

\renewcommand{\arraystretch}{1.4}
\begin{table}[tbp]
    \centering
    \begin{tabular}{clllllp{2.9cm}}
      \toprule
         &  Year
         & Paper
         & Space
         & Preprocessing
         & Query
         & Comments  \\
      \midrule \rotlong{8}{two point shortest path}
        & 1989 & \cite{ShortestPathSimplePolygon}
         & $n$
         & $n$
         & $\log n$
         & simple polygon  \\

         & 1999 & \cite{chiang99two_point_euclid_short_path_queries_plane}
         &  $n^{11}$
         & $n^{11}$
         & $\log n$
         & \\
         & 1999 & \cite{chiang99two_point_euclid_short_path_queries_plane}
         &  $n^{10}\log n$
         & $n^{10}\log n$
         & $\log^2n$
         &  \\

         & 1999 & \cite{chiang99two_point_euclid_short_path_queries_plane}
         &  $n^{5 + 10\delta + \eps}$
         & $n^{5+10\delta + \eps}$
         & $n^{1-\delta}\log n$
         & $0 < \delta \leq 1$ \\
         
                 & 1999  &\cite{chiang99two_point_euclid_short_path_queries_plane}
         & $n + h^5$
         & ?
         & $h \log n$
         &   \\

         & 2001 & \cite{chen2001geometric}
         & $n^2$
         & $n^2 \log n$
         & $q \log n $
         & $q = O(n)$ \\
         
                 & 2008 &\cite{guo08short_path_queries_polyg_domain}
         & $n^2$
         & $n^2 \log n$
         & $h\log n$
         &  \\

        & 2012  &\cite{bae12query}
         & $n^4\lambda_{66}(n)$
         & $n^4 \lambda_{65}(n) \log n$
         & $\log n$
         & query points on $\partial P$  \\\hline

         \rot{3}{single source} & 1993 & \cite{mitchell1993shortest}
         % & $Construction$
         & $n$
         & $n^{5/3}$
         & $\log n$
         &
         \\

        & 1999 & \cite{hershberger99optim_algor_euclid_short_paths_plane}
        &  $n$
        & $n\log n$
        &  $\log n$
        & \\

        & 2021 & \cite{wang21short_paths_among_obstac_plane_revis}
         % & $Construction$
         & $n$
         & $n\log n$
         & $\log n $
         & linear working space
         \\\hline

        \rot{2}{approx- imation} & 1995 & \cite{chen1995all}
         & $\frac{n}{\varepsilon} + n\log n$
         & $o(f^{3/2}) + \frac{n\log n}{\eps}$
         & $\frac{\log n}{\varepsilon} + \frac{1}{\varepsilon}$
         & Ratio $(6 + \varepsilon)$  \\

        & 2007 & \cite{thorup07compac_oracl_approx_distan_aroun_obstac_plane}
         & $\frac{n\log n}{\varepsilon} $
         & $\frac{n\log^3n}{\varepsilon^2} $
         & $\frac{1}{\varepsilon^3} + \frac{\log n}{\eps \log\log n}$
         & Ratio $(1+\varepsilon)$\\\hline

       \rot{2}{$L_1$- metric}& 2000 & \cite{chen2000shortest}
         & $n^2 \log n$
         & $n^2 \log^2n$
         & $\log^2 n$
         &  \\

       & 2020 & \cite{wang20divid_conquer_algor_two_point}
         & $n + \frac{h^2 \log^3 h}{ \log \log h}$
         & $n + \frac{h^2 \log^4 h}{ \log \log h}$
         & $\log n$
         &   \\

      \bottomrule
    \end{tabular}
    % \vspace{0.2cm}
    \caption{Overview of results on shortest paths in polygonal domains.
    All bounds are asymptotic.
    The parameter $h$ represents the number of holes.
    The parameter $q$ is the minimum of the number of vertices that
    $s$ or $t$ sees.
    The parameter $f$ is the minimum number of faces needed to cover the vertices in a certain planar graph.
    The function $\lambda_m(n)$ is the maximum length of a Davenport-Schinzel sequence of $n$ symbols of order $m$.  
  }
    \label{tab:related-work}
\end{table}

As mentioned before, when the domain is restricted to a simple polygon, there exists an optimal linear size-data
structure with $O(\log n)$ query time by Guibas and Hershberger~\cite{ShortestPathSimplePolygon}.

When we consider the algorithmic question of finding the shortest path
between two (fixed) points in a polygonal domain, the state-of-the-art algorithms
build the so-called shortest path map from the source
$s$~\cite{guibas1987linear, hershberger1994computing}.
Hershberger and Suri presented
such an $O(n)$-space data structure that can answer shortest path queries from a fixed point $s$ in $O(\log n)$ time~\cite{hershberger99optim_algor_euclid_short_paths_plane}.
The construction takes $O(n\log n)$ time and space.
This was recently improved by
Wang~\cite{Wang23euclid} to run in
optimal $O(n + h\log h)$ time and to use only $O(n)$
working space when a triangulation is given, where $h$ denotes the number of holes in the domain.

By parameterizing the query time by the number of holes $h$, Guo, Maheshwari, and  Sack~\cite{guo08short_path_queries_polyg_domain} manage to build a data structure that uses $O(n^2)$ space and has query time~$O(h \log n)$.

Bae and Okamoto~\cite{bae12query} study the special case where both query points are restricted to lie on the boundary $\partial P$ of the polygonal domain.
They present a data structure of size $O(n^4 \lambda_{66}(n)) \approx O(n^5)$ that can answer queries in $O(\log n)$ time.
 Here, $\lambda_m(n)$ denotes the maximum length of a Davenport–Schinzel sequence of order $m$
 on $n$ symbols~\cite{Davenport_schinzel_sequences}.

 Two other variants that were considered are using
 approximation~\cite{chen1995all,thorup07compac_oracl_approx_distan_aroun_obstac_plane},
 and using the $L_1$-norm~\cite{ chen2016two,
   chen2000shortest,wang20divid_conquer_algor_two_point}. Very
   recently, Hagedoorn and Polishchuk~\cite{hagedoorn23_link_distance}
   considered two-point shortest path queries with respect to the
   link-distance, i.e.\ the number of edges in the path. This seems to
   make the problem harder rather than easier: the space usage of the
   data structure is polynomial, and likely much larger than the
   geodesic distance data structures, but they do not provide an exact
   bound.

%%%%%%%%%%%%%%%%%%%%%%%%%%%%%%%%%%%%%%%%%
\subsection{Results}  
Our main result is the first improvement in more
than two decades that achieves optimal $O(\log n)$ query time.

\begin{restatable}[Main Theorem]{theorem}{main}
  \label{thm:2pt_shortest_path_query_data_structure}
  % Let $P$ be a polygonal domain with $n$ vertices.
  % For any constant $\eps >0$, we can build a data structure using $O(n^{10+\eps})$
  % space and expected preprocessing time that can answer two-point shortest path queries in $O(\log n)$ time. Alternatively, we can build a data structure using $O(n^{9+\eps})$ space and expected preprocessing time that can answer queries in $O(\log^2 n)$ time.
  For any constant $\eps>0$, we can build a data structure solving the \TwoPointShortestPath problem using  $O(n^{10+\eps})$ space and expected preprocessing time that has
  $O(\log n)$ query time.
  Alternatively, we can build a data structure using $O(n^{9+\eps})$ space and expected preprocessing time that has $O(\log^2 n)$ query time.
\end{restatable}
One of the main downsides of the two-point shortest path data structure
is the large space usage.
One strategy to mitigate the space usage is to allow for a larger
query time.
For instance, Chiang and Mitchell presented a myriad of different space-time trade-offs.
One of them being $O(n^{5+10\delta+\eps})$  space with $ O(n^{1-\delta}\log n)$ query time for $0 < \delta \leq 1$.
Our methods allow naturally for such a trade-off.
We summarize our findings in the following theorem.

\begin{restatable}{theorem}{tradeoff}
\label{thm:tradeoff}
  %   Let $P$ be a polygonal domain with $n$ vertices.
  %   For any constant $\eps >0$ and integer $1 \leq \ell \leq n$, we can build a data structure using $O(n^{9+\eps}/\ell^{4+O(\eps)})$
  % space and expected preprocessing time that can answer two-point shortest path queries in $O(\ell \log^2 n)$ time.
    For any constant $\eps >0$ and integer $1 \leq \ell \leq n$, we can build a data structure for the \TwoPointShortestPath problem using $O(n^{9+\eps}/\ell^{4+O(\eps)})$
  space and expected preprocessing time that has $O(\ell \log^2 n)$ query time.
\end{restatable}

For example, for $\ell = n^{3/4}/\log n$ we obtain an $O(n^{6 + \eps}\log^4n)$-size data structure with query time $O(n^{3/4}\log n)$, which improves the $O(n^{7.5 + \eps})$-size data structure with similar query time of~\cite{chiang99two_point_euclid_short_path_queries_plane}.

Another way to reduce the space usage is to restrict the problem setting.
With our techniques it is natural to consider the setting where either one or both of the query points
are restricted to lie on the boundary of the domain. In case we only restrict one of the query points to the boundary, we obtain the following result. Note that the other query point can lie anywhere in $P$.

\begin{restatable}{theorem}{boundary}
\label{thm:thm:2pt_boundary_shortest_path_query_data_structure}
  %Let $P$ be a polygonal domain with $n$ vertices. 
  For any constant $\eps >0$, we can build a data structure for the \TwoPointShortestPath problem in $O(n^{6+\eps})$
  space and expected time that can answer queries for $s \in \partial P$ and $t \in P$ in $O(\log^2 n)$ time.
\end{restatable}

When both query points are restricted to the boundary, we obtain the following result.

\begin{restatable}{theorem}{boundaryst}
\label{thm:thm:2pt_boundary_st_shortest_path_query_data_structure}
  %Let $P$ be a polygonal domain with $n$ vertices. 
  For any constant $\eps >0$, we can build a data structure for the \TwoPointShortestPath problem in $O(n^{4+\eps})$
  space and time that can answer queries for $s \in \partial P$ and $t \in \partial P$ in $O(\log n)$ time.
\end{restatable}

Note that the running time for this problem is deterministic. This improves the result by Bae and Okamoto~\cite{bae12query}, who provide an $O(n^4 \lambda_{66}(n)) \approx O(n^5)$-sized structure for this problem.

% In \textit{Hopcroft's problem} we are given $n$ points $P$ in the plane and $n$ lines $\L$. 
% We then ask whether there exist a line $\ell \in \L$ containing a point from $P$.
% See \Cref{fig:line-emptiness} for an illustration.
% We define further the integer version of Hopcroft's problem where each point in $P$ has integer coordinates bounded by $L$. In the same spirit, we ask for every line $\ell = \{(x,y)\in\R^n : ax+b = y \}$ that the slope and intercept ($a,b$) are integers as well with absolute value bounded by $L$.

% \begin{conjecture*}[Integer Hopcroft Lower Bound]
%     Let $L = O(n^c)$ then every algorithm to solve the integer Hopcroft problem needs at least $\Omega(n^2)$ time steps on a realRAM model of computation.
% \end{conjecture*}

% Note that there are various unconditional lower bounds for Hopcroft's problem.
% We refer to the article by Erickson~\cite{Jeff2000} for detailed lower bounds.
% Here, we only remark that existing lower bounds are usually not with integer coordinates nor are they with respect to the realRAM model.
% It is unclear to us if integer coordinates and the realRAM simplify Hopcroft's problem to the degree that
% we can gain actual algorithmic gains.

% We denote by $p$ the preprocessing time and by $t$ the query time.

To complement our positive algorithmic results, we also studied
lower bounds.
Unfortunately, we were only able to find a lower bound assuming that
the Integer version of Hopcroft's problem takes $\Omega(n^2)$ time 
to compute.
See \Cref{sec:lower_bounds} for a precise definition of Hopcroft's problem.

\begin{restatable}{theorem}{lowerbound}
\label{thm:lowerbound}
      Consider an data structure for the \TwoPointShortestPath problem with
    space complexity $\spacecomplexity$ and query time $\querytime$.
    Assuming the Integer Hopcroft lower bound, 
    if $\querytime$ is polylogarithmic then it holds that  
     $\spacecomplexity = \Omega(n^2)$.
\end{restatable}

\subsection{Discussion}
In this section, we briefly highlight strengths and limitations of
our results.

\subparagraph{Applications.} For many applications, our current data
structure is not yet efficient enough to improve the state of the
art. Yet, it is conceivable that further improvements to the two-point
shortest path data structure will trigger a cascade of improvements
for other problems in polygonal domains. One problem in which we do
already improve the state of the art is in computing the geodesic
diameter, that is, the largest possible (geodesic) distance between
any two points in $P$. Bae \etal~\cite{Bae13_geodesic_diameter} show
that with the two-point shortest path data structure of Chiang and
Mitchell, the diameter can be computed in $O(n^{7.73 + \eps})$
time. By applying our improved data structure (with $\ell = n^{2/5}$), the running time can
directly be improved to $O(n^{7.4+\eps})$. Another candidate
application is computing the geodesic center of $P$, i.e.\ a point that
minimizes the maximum distance to any point in
$P$. Wang~\cite{Wang18_geodesic_center} shows how to compute a center
in $O(n^{11} \log n)$ time using two-point shortest path
queries. Unfortunately, the two-point shortest path data structure is
not the bottleneck in the running time, so our new data structure does
not directly improve the result yet. Hence, more work is required
here.

\subparagraph{Challenges.} Data structures often use
divide-and-conquer strategies to efficiently answer queries. One of
the main challenges in answering shortest path queries is that it is
not clear how to employ such a divide and conquer strategy. We cannot
easily partition the domain into independent subpolygons (the strategy
used in simple polygons), as a shortest path may somewhat arbitrarily
cross the partition boundary. Furthermore, even though it suffices to find a
single vertex $v$ on the shortest path from $s$ to $t$ (we can then
use the shortest path map of $v$ to answer a query), it is hard to
structurally reduce the number of such candidate vertices. It is, for
example, easily possible that both query points see a linear number of
vertices of $P$, all of which produce a candidate shortest path of
almost the same length. Hence, moving $s$ or $t$ slightly may result
in switching to a completely different path. 

We tackle these
challenges by considering a set $\T$ of regions that come from the
triangulated shortest path maps of the vertices of $P$. The number of
regions in \T is a measure of the remaining complexity of the
problem. We can gradually reduce this number in a divide and conquer
scheme using cuttings. However, initially we now have $O(n^2)$ regions
as candidates to consider rather than just $n$ vertices. Surprisingly,
we show that we can actually combine this idea with a notion of
\emph{relevant pairs of regions}, of which there are only $O(n)$. This
then allows us to use additional tools to keep the query time and
space usage in check. It is this careful combination of these ideas
that allows us to improve the space bound of Chiang and
Mitchell~\cite{chiang99two_point_euclid_short_path_queries_plane}.

\subparagraph{Lower bounds.}
An important question is how much improvement of the space complexity is actually possible while retaining polylogarithmic query time. To the best of our knowledge there exist no non-trivial lower bounds on the space complexity of a two-point shortest path data structure. In \Cref{sec:lower_bounds}, we show a conditional quadratic lower bound. 
%Computing a shortest path is at least as hard as ray shooting among line segments~\cite{Agarwal_ray_shooting}, which is comparable in difficulty to simplex range searching. Since simplex range searching has a roughly quadratic lower bound~\cite{chazelle_simplex_lower_bounds}, we expect two-point shortest path queries to have a similar lower bound. 
We conjecture that it may even be possible to obtain a super-quadratic lower bound. An indication for this is that the equivalence decomposition of $P$, which is the subdivision of $P$ into cells such that the shortest path maps are topologically equivalent, has complexity $\Omega(n^4)$~\cite{chiang99two_point_euclid_short_path_queries_plane}. However, this does not directly imply a lower bound on the space of any two-point shortest path data structure. We leave proving such a bound as an exciting open problem.

\subsection{Organization}
In Section~\ref{sec:overview}, we give an overview of our main data
structure, including the special case where one of the query points
lies on the boundary. In
Section~\ref{sec:Decomposing_the_Distance_Computation}, we show how to
decompose the distance computation, such that we can apply a
divide-and-conquer approach to the problem. Cuttings are a key
ingredient of our approach, we formally define them, and prove some
basic results about them in Section~\ref{sec:Cuttings}. In
Sections~\ref{sec:A_data_structure_for_when_Rs_and_Rt_are_given} to~\ref{sec:A_two-point_shortest_path_data_structure}, we build up our main data structure step-by-step. We first consider the subproblem where a subset of regions for both $s$ and $t$ is given in Section~\ref{sec:A_data_structure_for_when_Rs_and_Rt_are_given}, then we solve the subproblem where a subset of regions is given only for $s$ in Section~\ref{sec:given_Rs}, and finally we tackle the general two-point shortest path problem in Section~\ref{sec:A_two-point_shortest_path_data_structure}. In Section~\ref{sec:Space-_time_trade-off}, we generalize this result to allow for a space-time trade-off. In Section~\ref{sec:A_data_structure_for_$s$_on_the_boundary}, we discuss our results for the restricted problem where either one or both of the query points must lie on the boundary of $P$. Lastly, in Section~\ref{sec:lower_bounds}, we discuss the $\Omega(n^2)$ space lower bound based on the Integer Hopcroft lower bound.

\section{Proof Overview}\label{sec:overview}

\subparagraph*{Direct visibility.}
As a first step, we build the visibility complex as described by Pocchiola and Vegter~\cite{pocchiola96the_visib_compl}. 
It allows us to query in $O(\log n)$ time if $s$ and $t$ can see each other.
If so, the
line segment connecting them is the shortest path. 
The visibility complex uses
$O(n^2)$ space and can be built in $O(n^2)$ time.
So, in the remainder, we assume that $s$ and $t$ cannot see each other, hence
their shortest path will visit at least one vertex of $P$.

\begin{figure}[tbp]
  \centering
  \includegraphics{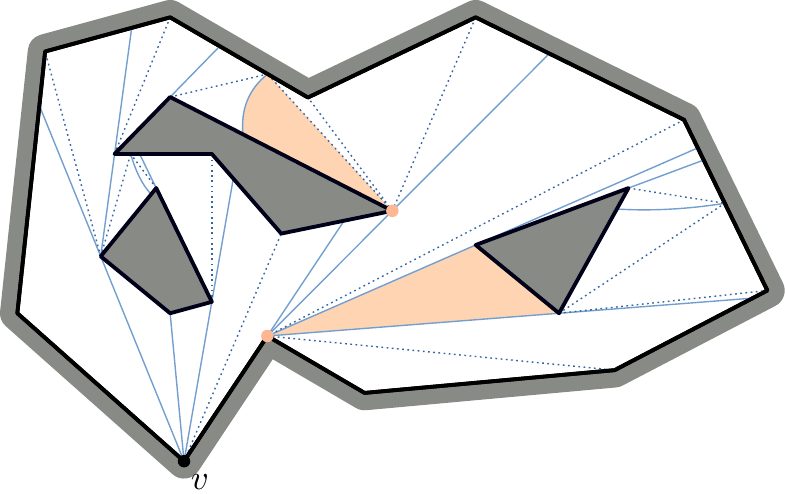}
  \caption{The augmented shortest path map of a vertex $v$. The shortest path map edges are solid, and the additional edges in the augmented shortest path map are dotted. Each region is bounded by three curves, of which at least two are line segments. Two regions and their apices are highlighted.}
  \label{fig:spm}
\end{figure}

\subparagraph*{Augmented shortest path maps.}
In our approach, we build a data structure on the regions provided by the \emph{augmented shortest path maps} of all vertices of $P$. The shortest path map of a point $p \in P$ is a partition of $P$ into maximal regions, such that for every point in a region $R$ the shortest path to $p$
traverses the same vertices of $P$~\cite{hershberger99optim_algor_euclid_short_paths_plane}. 
To obtain the augmented shortest path map $\SPM(p)$, we refine the shortest path map by connecting each boundary vertex of a region $R$ by a line segment with the apex $v_R$ of that region, i.e.\ the first vertex on the shortest path from any point in $R$ towards $p$.
See Figure~\ref{fig:spm} for an example. 
All regions in $\SPM(p)$ are ``almost'' triangles; they are bounded by three curves, two of which are line segments, and the remaining is either a line segment or a piece of a hyperbola.
The (augmented) shortest path map has complexity $O(n)$ and can be constructed in $O(n \log n)$ time~\cite{hershberger99optim_algor_euclid_short_paths_plane}. 
Let \T be the multi-set of \emph{all} augmented shortest path map regions of \emph{all} the vertices of $P$. 
As there are $n$ vertices in $P$, there are $O(n^2)$ regions in $\T$.

Because we are only interested in shortest paths that contain at least one vertex, the shortest path between two points $s,t \in P$ consists of an edge from $s$ to some vertex $v$ of $P$ that is visible from $s$, a shortest path from $v$ to a vertex $u$ (possibly equal to $v$) that is visible from $t$, and an edge from $u$ to $t$. For two regions $S,T \in \T$ with $s \in S$ and $t \in T$, we define
$f_{ST}(s,t) = ||sv_S|| + d(v_S,v_T) + ||v_Tt||$. 
The distance $d(s,t)$ between $s$ and $t$ is realized by this function when $v_S = v$ and $v_T = u$. 
As for any pair $S,T$ with $s \in S$ and $t \in T$ the function $f_{ST}(s,t)$ corresponds to the length of some path between $s$ and $t$ in $P$, we can obtain the shortest distance by taking the minimum over all of these functions. See Figure~\ref{fig:minimum_f}, and refer to Section~\ref{sec:Decomposing_the_Distance_Computation} for the details.  
In other words, if we denote by $\T_p$ all regions that contain a point $p \in P$, we have
\begin{equation*}
    d(s,t) = \min \{ f_{ST}(s,t) : S \in \T_s,\hspace{1pt} T \in \T_t\}.
\end{equation*}

\begin{figure}[tbp]
  \centering
  \includegraphics{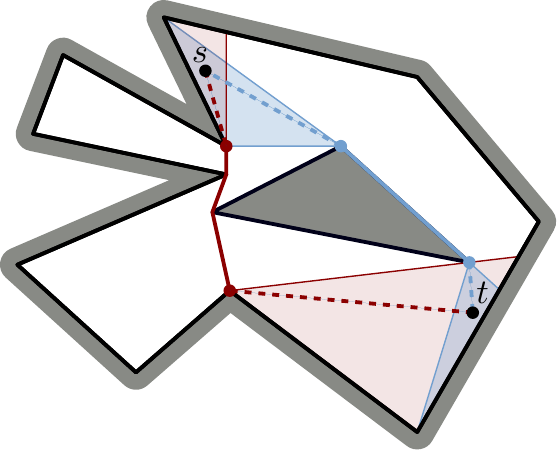}
  \caption{Two pairs of relevant regions in red and blue with the path whose length is $f_{ST}(s,t)$.}
  \label{fig:minimum_f}
\end{figure}

\subparagraph*{Lower envelope.}
Given two multi-sets $\A,\B \subseteq \T$, we want to construct a data structure that we can efficiently query at any point $(s,t)$ with $s \in \bigcap \A$ and $t \in \bigcap \B$ to find $\min\{f_{ST}(s,t): S\in \A, T \in \B\}$. We refer to this as a \LowerEnvelope data structure. We can construct such a data structure of size $O(\min \{\sizeA,\sizeB,n\}^{6+\eps})$ with $O(\log (\min\{\sizeA,\sizeB,n\}))$ query  time, or of size $O(\min \{\sizeA,\sizeB,n\}^{5+\eps})$  with $O(\log^2 (\min\{\sizeA,\sizeB,n\}))$ query time, as follows.

The functions $f_{ST}$ are four-variate algebraic functions of constant degree. Each such function gives rise to a surface in $\R^5$, which is the graph of the function~$f$.
Koltun~\cite{koltun04almos} shows that the vertical decomposition of $m$ such surfaces in $\R^5$ has complexity $O(m^{6+\eps})$, and can be stored in a data structure of size $O(m^{6+\eps})$ so that we can query the value of the lower envelope, and thus $d(s,t)$, in $O(\log m)$ time.  
More recently, Agarwal~\etal~\cite{Agarwal21_polynomial_partitioning} showed that a collection of $m$ semialgebraic sets in $\R^5$ of constant complexity can be stored using $O(m^{5 + \eps})$ space such that vertical ray-shooting queries, and thus lower envelope queries, can be answered in $O(\log^2 m)$ time.

We limit the number of functions $f_{ST}(s,t)$ by using an observation of
Chiang and Mitchell~\cite{chiang99two_point_euclid_short_path_queries_plane}. They note that we do not need to consider all pairs $S\in \A, T\in \B$, but only $\min\{\sizeA,\sizeB,n\}$ relevant pairs. 
Two regions form a \emph{relevant} pair if they belong to the same augmented shortest path map $\SPM(v)$ of some vertex $v$.
(To be specific, if $v$ is any vertex on the shortest path from $s$ to $t$, 
then the minimum is achieved for $S$ and $T$ in the shortest path map of $v$.) We thus obtain a \LowerEnvelope data structure by constructing the data structure of~\cite{Agarwal21_polynomial_partitioning} or~\cite{koltun04almos} on these $\min\{\sizeA,\sizeB,n\}$ functions.

Naively, to build a data structure that can answer shortest path queries for any pair of query points $s,t$, 
we would need to construct this data structure for all possible combinations of $\T_s$ and $\T_t$. 
The overlay of the $n$ augmented shortest path maps has worst-case complexity $\Omega(n^4)$~\cite{chiang99two_point_euclid_short_path_queries_plane}, 
which implies that we would have to build $\Omega(n^8)$ of the
\LowerEnvelope data structures. 
Indeed, this results in an $O(n^{14 + \eps})$-size data structure, 
and is one of the approaches Chiang and Mitchell consider~\cite{chiang99two_point_euclid_short_path_queries_plane}.
Next, we describe how we use cuttings to reduce the number of \LowerEnvelope data structures we construct.

\subparagraph{Cutting trees.}
Now, we explain how to determine $\T_s$ more efficiently using cuttings and cutting trees.
Suppose we have a set $\A$ 
of $N$ (not necessarily disjoint) triangles in the plane. 
A \emph{$1/r$-cutting} $\Xi$ of $\A$ is then a subdivision of the plane into constant complexity
cells, for example triangles, such that each cell $\Delta$ in $\Xi$ is
intersected by the boundaries of at most $N/r$ triangles in $\A$~\cite{Chazelle93_Cutting}. 
There can thus still be many triangles that fully contain a cell, 
but only a limited number whose boundary intersects a cell.
In our case, the regions in $\T$ are \emph{almost} triangles, called \emph{Tarski cells}~\cite{agarwal94range_searc_semial_sets}. 
See Section~\ref{sec:tarksi-cell} for a detailed description of Tarski cells.
As we explain in Section~\ref{sec:Cuttings}, we can always construct such a cutting with only $O(r^2)$ cells for these types of regions efficiently.

Let $\Xis$ be a $1/r$-cutting of $\T$. For $s\in\Delta\in\Xis$ the regions $R\in \T$ that fully contain $\Delta$ also contain $s$.
To be able to find the remaining regions in $\T_s$, we recursively build cuttings on the $N/r$ regions whose boundary intersects $\Delta$.
This gives us a so-called cutting tree.
By choosing $r$ appropriately, the cutting tree has constantly many levels.
The set $\T_s$ is then the disjoint union of all regions obtained in a root-to-leaf path in the cutting tree.
Note that using a constant number of point location queries it is possible to find all of the vertices on this path.

\subparagraph*{The multi-level data structure.} Our data structure, which we describe in detail in Sections~\ref{sec:given_Rs} and~\ref{sec:A_two-point_shortest_path_data_structure}, is essentially a multi-level cutting tree, as in~\cite{Clarkson87_new_applications_random_sampling}.
See Figure~\ref{fig:overview} for an illustration.
The first level is a cutting tree that is used to find the regions that contain~$s$, as described before. For each cell $\Delta$ in a cutting $\Xis$, we construct another cutting tree to find the regions containing $t$. Let $\A$  be the set of regions fully containing~$\Delta$ and $|\A| = k$. Then, the second-level cutting $\Xit$ is built on the $O(kn)$ candidate regions that are in a relevant pair. See Figure~\ref{fig:subproblem}. We process the regions intersected by a cell in $\Xit$ recursively to obtain a cutting tree.
Additionally, for each cell $\Delta' \in \Xit$, we construct the \LowerEnvelope data structure on the sets $\A,\B$, where $\B$ is the set of regions that fully contain $\Delta'$. This allows us to efficiently obtain $\min f_{ST}(s,t)$ for $S\in\A$ and $T \in \B$. 

\begin{figure}
    \centering
    \includegraphics{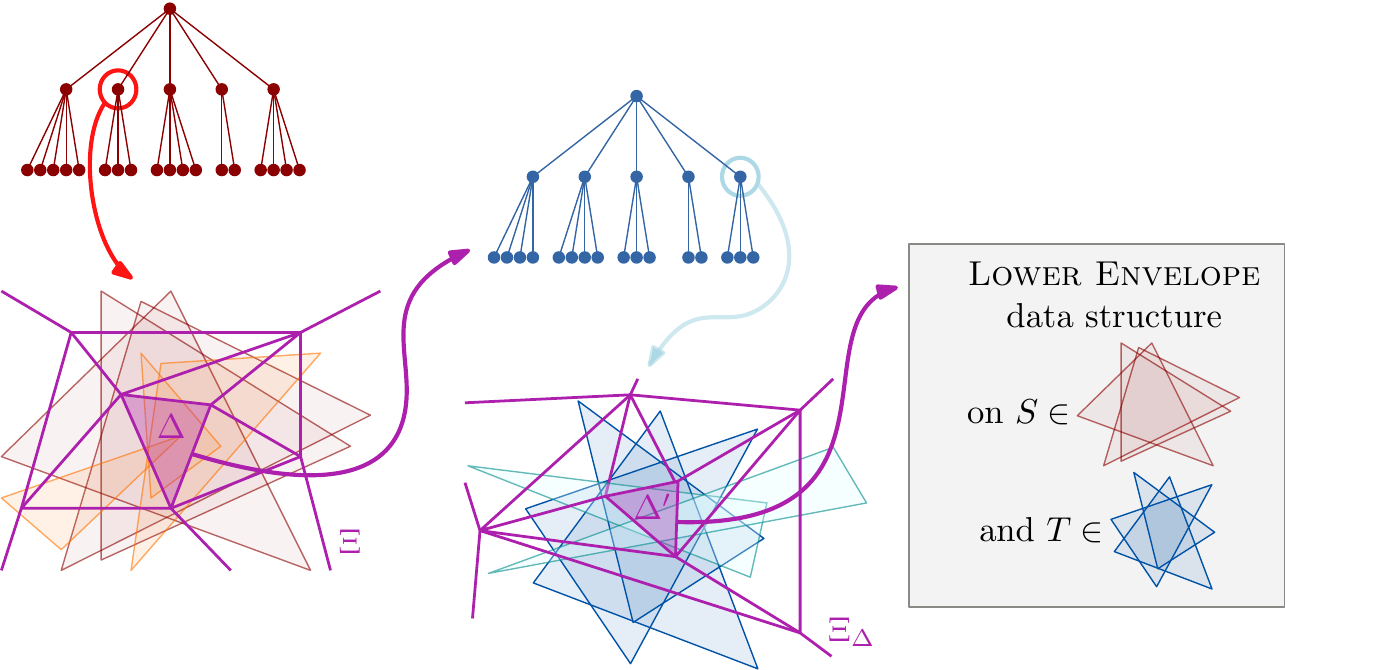}
    \caption{Overview of our data structure. The first level cutting tree (red) is built by recursively constructing a cutting $\Xis$ on the (orange) regions that intersect a cell $\Delta$ (purple). For each cell $\Delta$, we store a second level cutting tree (blue). For each cell $\Delta'$ in $\Xit$, we build a \LowerEnvelope data structure on all regions that fully contain $\Delta$ (dark red) and $\Delta'$ (dark blue).}
    \label{fig:overview}
\end{figure}

\begin{figure}[tb]
  \centering
  \includegraphics{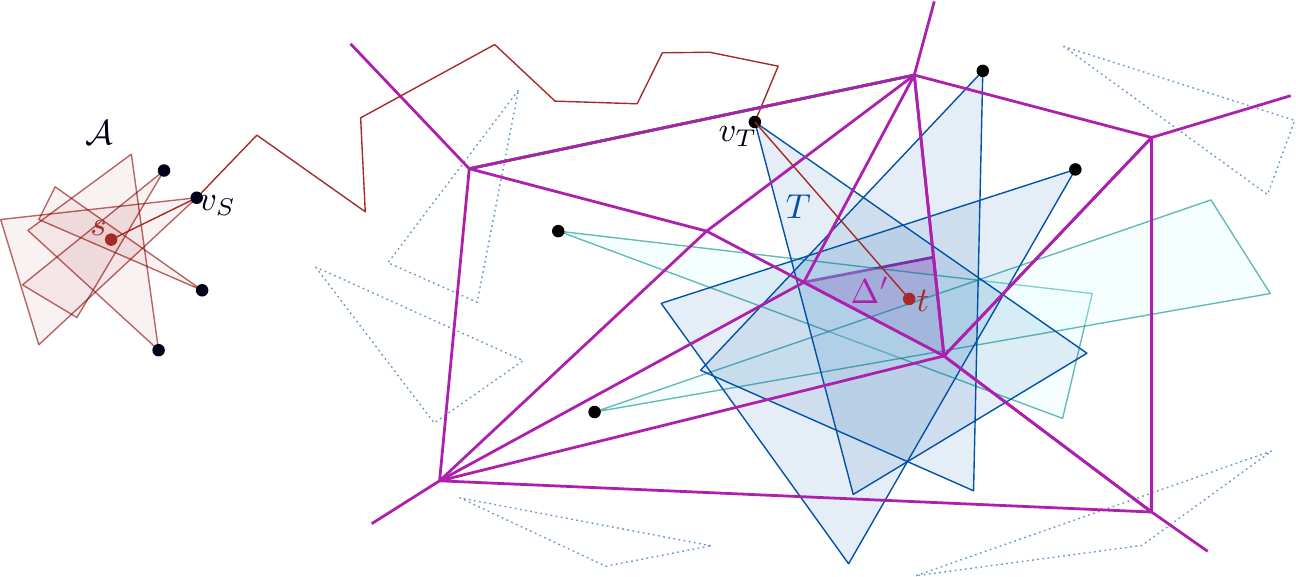}
  \caption{A sketch of the subproblem considered here, computing
    $\min_{S \in \A, T \in \T}f_{ST}(s,t)$. We build a $1/r$-cutting $\Xit$
    (shown in purple) on the set of relevant regions in $\T$ (blue). The regions
    $\T_t \subseteq \T$ either fully contain the
    cell $\Delta' \in \Xit$ of the cutting that contains $t$ (dark blue), or their
    boundaries intersect $\Delta'$ (light blue).   }
  \label{fig:subproblem}
\end{figure}

\subparagraph{Queries.}  To query our data structure with two points
$s,t$, we first locate the cell $\Delta_s$ containing $s$ in the
cutting $\Xis$ at the root. We compute $\min f_{ST}(s,t)$ for all regions $S$
that intersect $\Delta_s$, but do not fully contain $\Delta_s$, by recursively querying the child node corresponding to
$\Delta_s$. To compute $\min f_{ST}(s,t)$ for all $S$ that fully
contain $\Delta_s$, we query its associated data structure. To this
end, we locate the cell $\Delta_t$ containing $t$ in $\Xi_{\Delta_s}$,
and use its lower envelope structure to compute $\min f_{ST}(s,t)$
over all $S$ that fully contain $\Delta_s$ and all $T$ that fully
contain $\Delta_t$. We recursively query the child corresponding to
$\Delta_t$ to find $\min f_{ST}(s,t)$ over all $T$ that intersect~$\Delta_t$.

\subparagraph{Sketch of the analysis.}
By choosing $r$ as $n^{\delta}$ for some constant $\delta = O(\eps)$, we can achieve that each cutting tree has only constant height. The total query time is thus $O(\log n)$, when using the \LowerEnvelope data structure by Koltun~\cite{koltun04almos}.
Next, we sketch the analysis to bound the space usage of the first-level cutting tree, under the assumption that a second-level cutting tree, including the \LowerEnvelope data structures, uses $O(n^2 \min\{k,n\}^{6+\eps})$ space (see Section~\ref{sec:given_Rs}). 

To bound the space usage, we analyze the space used by the \emph{large} levels, where the number of regions is greater than $n$, and the \emph{small} levels of the tree separately, see Figure~\ref{fig:analysis}. 
There are only $O(n^2)$ large nodes in the tree. For these $\min\{k,n\} = n$, so  each stores a data structure of size $O(n^{8+\eps})$. For the small nodes, the size of the second-level data structures decreases in each step, as $k$ becomes smaller than $n$. Therefore, the space of the root of a small subtree, which is $O(n^{8+\eps})$, dominates the space of the other nodes in the subtree. As there are $O(n^2)$ small root nodes, the resulting space usage is $O(n^{10 + \eps})$.

\begin{figure}
    \centering
    \includegraphics{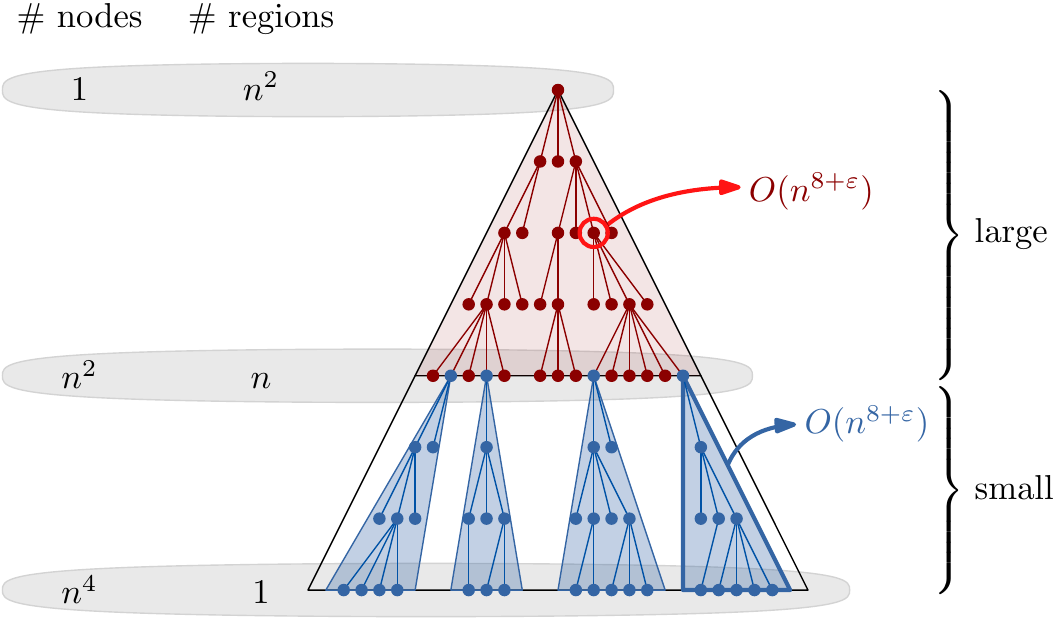}
    \caption{We analyze the large levels, built on $\geq n$ regions, and the small levels, built on $< n$ regions, separately. The total space usage is $O(n^{10+\eps})$.}
    \label{fig:analysis}
\end{figure}

\subparagraph{A space-time trade-off.} We can achieve a trade-off between the space usage and the query time by grouping the polygon vertices, see Section~\ref{sec:Space-_time_trade-off} for details. We group the vertices into $\ell$ groups, and construct our data structure on each set of $O(n^2/\ell)$ regions generated by a group. This results in a query time of $O(\ell\log^2 n)$ and a space usage of $O(n^{9+\eps}/\ell^{4+O(\eps)})$ when we apply the vertical ray-shooting \LowerEnvelope data structure of Lemma~\ref{lem:vertical_ray_shooting}.

\subparagraph{A data structure for $s$ or $t$ on the boundary.} In Section~\ref{sec:A_data_structure_for_$s$_on_the_boundary}, we show how to adapt our data structure to the case where one (or both) of the query points, say $s$, is restricted to lie on the boundary of the domain. The main idea is to use the same overall approach as before, but consider a different set of regions for $s$ and $t$. The set of candidate regions for $s$ now consists of intervals formed
by the intersection of the \SPM regions with the boundary of $P$. Instead of building a cutting tree on these regions, we build a segment tree, where nodes again store a cutting tree on the regions for $t$. As the
functions $f_{ST}(s,t)$ are now only three-variate rather than four-variate, the resulting space usage is only $O(n^{6+\eps})$. 

When both query point are restricted to the boundary, the set of candidate regions for both $s$ and $t$ consists of intervals along $\partial P$. In this case, a node of the segment tree stores another segment tree on the regions for $t$. Because the functions $f_{ST}(s,t)$ are only bivariate, the space usage is reduced to $O(n^{4+\eps})$.

\subsection{Conclusion}\label{sec:concluding_remarks}

\subparagraph{Improving the space bound.} Both \LowerEnvelope data structures we use are actually more powerful
than we require: one allows us to perform point location queries in the
vertical decomposition of the \emph{entire} arrangement, and the other allows us to perform vertical ray-shooting from \emph{any} point in the arrangement.
While we are only
interested in lower envelope queries, i.e.\ vertical ray-shooting from a single plane. The (projected) lower
envelope of $m$ four-variate functions has a complexity of only
$O(m^{4+\eps})$~\cite{Sharir94lower_envelopes}. However, it is unclear if we can store this lower
envelope in a data structure of size $O(m^{4+\eps})$ while retaining
the $O(\log m)$ query time.
This would immediately reduce the space of our data structure to~$O(n^{8+\eps})$.

%%%%%%%%%%%%%%%%%%%%%%%%%%%%%%%%
\section{Decomposing the distance computation}
\label{sec:Decomposing_the_Distance_Computation}
%%%%%%%%%%%%%%%%%%%%%%%%%%%%%%%%

As described in Section~\ref{sec:overview}, we assume that there is at least one vertex on the shortest path between two query points. We will decompose the distance computation using the regions from the \emph{augmented shortest path maps} of all vertices.
Let $\SPM'(p)$ be the \emph{shortest path map} of source point $p$ in
$P$, i.e.\ a partition of $P$ into maximal closed regions such that for every
point in such a region $R$ the shortest path to $p$ visits the same
sequence of polygon
vertices~\cite{hershberger99optim_algor_euclid_short_paths_plane}. 
Let $v_R$ be the first vertex on the shortest path from any point in $R$
towards $p$. We refer to $v_R$ as the \emph{apex} of region $R$. We
``triangulate'' every region $R$ of $\SPM'(p)$ by connecting the
boundary vertices of $R$ with the apex $v_R$. The resulting
subdivision of $P$ is the \emph{augmented shortest path map}
$\SPM(p)$ of $P$ with respect to source point $p$. All regions in
$\SPM(p)$ consist of at most three curves, two of which are line
segments, and the remaining edge is either a line segment or a piece of a hyperbola, see Figure~\ref{fig:spm}. The (augmented) shortest path map has
complexity $O(n)$~\cite{hershberger99optim_algor_euclid_short_paths_plane}.

Consider the augmented shortest path maps for all vertices of $P$, and
let $\T$ denote the multi-set of all such
regions. Hence, \T
consists of $O(n^2)$ regions. We associate with each region in $\T$ the vertex of $P$ that generated the region. 
Let $\T_s \subset \T$ be
the subset of regions that contain a point $s \in P$.

For a pair $S,T \in \T \times \T$ we define
$f_{ST}(s,t) = \|sv_S\| + d(v_S,v_T) + \|v_Tt\|$ for $s \in S$ and $t \in T$, see Figure~\ref{fig:minimum_f}. Observe that
$f_{ST}$ is essentially a 4-variate algebraic function of constant
complexity (since $d(v_S,v_T)$ is constant with respect to $s$ and
$t$). For a pair of points $s,t$ that cannot see each other we then
have that

\begin{equation}
  \label{eq:shortest}
  d(s,t)  = \min_{S \in \T_s, T \in \T_t} \|sv_S\| + d(v_S,v_T) +
   \|v_Tt\| = \min \{ f_{ST}(s,t) : S \in \T_s,\hspace{1pt} T \in \T_t\} .
\end{equation}

Note that whenever either of the two query points lies on a polygon vertex, we can directly answer the query by considering the shortest path map of that vertex. 

\subparagraph{Relevant region-pairs.} 
We say that a pair of regions
$(S,T)$ is \emph{relevant} if and only if they appear in a single
augmented shortest path map $\SPM(v)$ for some vertex $v$. We say a region is relevant if it occurs in some relevant pair. Let
$\Rel(\X) = \{(S,T) \mid (S,T) \in \X \land (S,T) \text{ is relevant}\}$ denote the subset of relevant pairs of regions from a set of pairs $\X$. Chiang
and Mitchell~\cite{chiang99two_point_euclid_short_path_queries_plane}
observed that $d(s,t)$ is defined by a relevant pair. Let
$\L^\X(s,t)= \min \{f_{ST}(s,t) : (S,T) \in \Rel(\X) \}$. We thus have:

\begin{observation}[Chiang
  and
  Mitchell~\cite{chiang99two_point_euclid_short_path_queries_plane}]
  \label{obs:there_is_an_optimal_relevant_pair}
  There exists a pair $(S,T) \in \T_s \times \T_t$ such that: (i) the
  pair $(S,T)$ is relevant, and (ii) $d(s,t) = f_{ST}(s,t)$. Hence,
  \[ d(s,t) = \L^{\T_s \times \T_t}(s,t)
            = \min \{ f_{ST}(s,t) : (S,T) \in \Rel(\T_s \times \T_t)\}.
  \]
\end{observation}

\begin{proof}
  Let $\geod(s,t)$ be an optimal path, and let $v$ be the first vertex
  of $\geod(s,t)$. Pick $S$ and $T$ as regions in $\SPM(v)$ containing
  $s$ and $t$, respectively.
\end{proof}

Note that in this proof we can choose $S$ as any region in $\SPM(v)$ that contains $s$. Thus, if $s$ lies on the boundary of multiple regions, any of these regions will achieve the optimal distance. This implies that the set $\T_s$ can be limited to only one region from each shortest path map, so $|\T_s| = O(n)$, and similarly $|\T_t| = O(n)$.

\begin{lemma}
  \label{lem:size_rels}
  Let $\A \subseteq \T_s$, $\B \subseteq \T_t$ be two sets of regions. The set
  $\Rel(\A \times \B)$ contains at most
  $m = O(\min \{\sizeA,\sizeB \}) \leq O(n)$ pairs of regions.
\end{lemma}

\begin{proof}
    As stated before, $\T_s$, and
  thus \A, contains at most $O(n)$ regions. The same applies for
  $\T_t$ and \B, hence $\sizeA$ and $\sizeB$ are $O(n)$.

  Assume without loss of generality that $\sizeA \leq \sizeB$. What remains
  to show is that $\Rel(\A \times \B)$ contains at most $O(\sizeA)$
  regions. We charge every pair $(S,T) \in \Rel(\A \times \B)$
  to $S$, and argue that each region $S$ can be charged at most a
  constant number of times. Let $v$ be the vertex associated with $S$, i.e.\ the vertex that generated the region. Consider all
  regions $T_1,\dots,T_z$ such that
  $(S,T_i) \in \Rel(\A \times \B)$. All of these regions must
  appear in $\SPM(v)$ and contain $t$. There are
  at most $O(1)$ such regions. Hence, it follows that each region $S$
  is charged at most $O(1)$ times, and there are thus at most $O(\sizeA)$
  pairs of regions in $\Rel(\A \times \B)$.
\end{proof}

Furthermore, we observe that computing $\L^{\A \times \B}(s,t)$ is
decomposable:

\begin{lemma}
  \label{lem:decompose}
  Let $\A, \B \subseteq \T$ be sets of regions, let $\A_1,\dots,\A_k$ be
  a partition of $\A$, and let $\B_1,\dots,\B_\ell$ be a partition of
  $\B$. We have that
  \[
    \L^{\A \times \B}(s,t)
           = \min \{\L^{\A_i\times \B_j}(s,t) : 
           i \in \upto{1}{k}, j \in \upto{1}{\ell}\}
           .
  \]
\end{lemma}

\begin{proof}
  Recall that
  $\L^{\A \times \B}(s,t) =
  \min \{f_{ST}(s,t) : (S,T) \in \Rel(\A \times \B)\} $. Since computing
  the minimum is decomposable, all that we have to argue is that
  $\Rel(\A \times \B) = \bigcup_{i \in \upto{1}{k}, j \in \upto{1}{\ell}} \Rel(\A_i
  \times \A_j)$. By definition of $\Rel$ we have
  \begin{align*}
    \Rel((\mathcal{X} \cup \mathcal{Y}) \times \mathcal{Z}) &= \{(S,T) \mid (S,T) \in (\mathcal{X}
                                       \cup \mathcal{Y}) \times \mathcal{Z} \land (S,T)
                                       \text{ is relevant }\}\\
                                     &= \{(S,T) \mid (S,T) \in \mathcal{X} \times \mathcal{Z} \land (S,T)
                                       \text{ is relevant }\}\\
                                     &\cup
                                       \{(S,T) \mid (S,T) \in \mathcal{Y} \times \mathcal{Z} \land (S,T)
                                       \text{ is relevant }\}\\
                                       &= \Rel(\mathcal{X} \times \mathcal{Z}) \cup \Rel(\mathcal{Y} \times \mathcal{Z}),
  \end{align*}
  and symmetrically $\Rel(\mathcal{X} \times ( \mathcal{Y} \cup \mathcal{Z})) = \Rel(\mathcal{X} \times \mathcal{Y}) \cup \Rel(\mathcal{X} \times \mathcal{Z})$.
  The lemma now follows from repeated application of these equalities.
\end{proof}

By Observation~\ref{obs:there_is_an_optimal_relevant_pair} we have
that $d(s,t)=\L^{\T_s \times \T_t}(s,t)$. Applying
Lemma~\ref{lem:decompose} to $\T_s$ and $\T_t$, thus tells us that
we can compute $d(s,t)$ using a divide and conquer approach. We will use
cuttings to do so.

%%%%%%%%%%%%%%%%%%%%%%%%%%%%%%%%%
\section{Cuttings}
\label{sec:Cuttings}
%%%%%%%%%%%%%%%%%%%%%%%%%%%%%%%%%

In this section we first introduce Tarski cells. We then define a $1/r$-cutting on a set of Tarski cells, and prove that such a cutting exists and can be computed efficiently.
%%%%%%%%%%%%%%
\subsection{Tarski cells}
\label{sec:tarksi-cell}
%%%%%%%%%%%%%%

A \emph{Tarski cell} is a region of $\R^2$ bounded by a constant
number of semi-algebraic curves of constant
degree~\cite{agarwal94range_searc_semial_sets}. Observe that the regions generated by the shortest path maps are thus Tarski cells. A \emph{pseudo
  trapezoid} is a Tarski cell that is bounded by at most two
vertical line segments and at most two $x$-monotone curves. Observe
that a pseudo trapezoid itself is thus also $x$-monotone. Given a set
\T of $N$ of Tarski cells, let $\A(\T)$ denote their arrangement, and
let $\VD(\T)$ denote the \emph{vertical decomposition} of $\A(\T)$ 
into pseudo trapezoids. It is well known that both $\A(\T)$ and
$\VD(\T)$ have complexity $O(N^2)$ (as any pair of curves intersects
at most a constant number of times, and each curve has at most a
constant number of points at which it is tangent to a vertical
line). 

\subsection{Cuttings on Tarski cells}

Let \T be a set of $N$ regions, each a Tarski
cell, and let $r \in \upto{1}{N}$ be a parameter. 
Before we can introduce cuttings, we need to define conflict lists.
We say, a region $T \in \T$ \emph{conflicts} with a cell
$\nabla \subseteq \R^2$ if the boundary of $T$ intersects the interior
of $\nabla$. The \emph{conflict list} $\C_\nabla$ of $\nabla$ (with
respect to \T) is the set of all region $T \in \T$ that conflict with
it.

A \emph{$1/r$-cutting}
$\Xi$ of \T is a partition of $\R^2$ into Tarski cells, each of which
is intersected by the boundaries of at most $N/r$ regions in \T,
i.e.\ has a conflict list of size at most $N/r$~\cite{Chazelle93_Cutting}. The
following two Lemmas summarize that such cuttings exist, and can be
computed efficiently.

\begin{lemma}
  \label{lem:vertical_decomposition}
  Let \T be a set of $N$ Tarski cells, and let $c$ be a constant. For
  a random subset $\RR \subseteq \T$ of size $r$, we can compute $\VD(\RR)$, and all conflict lists in
  expected $O(Nr)$ time. Furthermore, we have that
  \[
    \Exp\left[ \sum_{\nabla \in \VD(\RR)} |\C_\nabla|^c  \right] =
    O(r^2(N/r)^c).
  \]
\end{lemma}

\begin{proof}
  The bound on the expected value of the quantity
  $W^c_i = \sum_{\nabla \in \VD(\RR)} |\C_\nabla|^c$ follows from a
  Clarkson-Shor type sampling
  argument~\cite{clarkson89applic_ii}. Therefore, we can compute
  $\VD(\RR)$ and its conflict lists by running the first $r$ rounds of
  a randomized incremental construction algorithm. 
  For completeness,
  we include the detailed argument here. Our presentation is based on that of
  Har-Peled~\cite{harPeled2011geometric}.

  \subparagraph{Algorithm.} We build the vertical decomposition
$\VD(\RR)$ by running a randomized incremental construction
algorithm~\cite{harPeled2011geometric} to construct $\VD(\T)$ for only
$r$ rounds.  Let $T_1,\dots,T_r$ denote the regions in the random order,
and let $\RR_i=\{T_1,\dots,T_i\}$ be the subset containing the first $i$
regions. So $\RR=\RR_r$. We maintain a bipartite conflict graph during the algorithm;
the sets of vertices are the pseudo-trapezoids in $\VD(\RR_i)$ and the
not-yet inserted regions from $\T \setminus \RR_i$. The conflict graph
has an edge $(\nabla,T)$ if and only if $T \in \C^\T_\nabla$. Using
this conflict graph, it is straightforward to insert a new region $T_i$
into $\VD(\RR_{i-1})$; the conflict graph tells us exactly which
pseudo trapezoids from $\VD(\RR_{i-1})$ should be deleted; each of
which is replaced by $O(1)$ new pseudo-trapezoids.

Let $E_i$ be the total number of newly created edges in the conflict
graph in step $i$ (i.e.\ the total number of new entries in the
conflict lists), and let $W_i = W^1_i$ denote the total size of all
conflicts of the pseudo-trapezoids in $\VD(\RR_i)$. The time required
by step $i$ of the algorithm is linear in $E_i$. We will argue that
$\Exp[E_i] = O(N)$, and thus the total expected running time is
$O(Nr)$.

\subparagraph{Analyzing $\Exp[E_i]$.} For a particular pseudo
trapezoid $\nabla \in \VD(\RR_i)$, the probability that $\nabla$ was
created by inserting region $T_i$ is at most $4/i$ (since $T_i$ must
contribute to one of the four boundaries of $\nabla$). Therefore, the
expected size of $E_i$, conditioned on the fact that the vertical
decomposition after $i$ rounds is $\VD(\RR_i)$, is
$\Exp[E_i \mid \VD(\RR_i) ] = O(\Exp[W_i]/i)$. Using the law of total
expectation, it then follows that
$\Exp[E_i] = \Exp[\Exp[W_i \mid \VD(\RR_i)]] = O(\Exp[W_i]/i)$. As we
argue next, we have that $\Exp[W^c_i]=O((N/i)^c i^2)$, and therefore
$\Exp[W_i]=\Exp[W^1_i]=O(Ni)$. In turn, this gives us
$\Exp[E_i] = O(Ni/i) = O(N)$ as claimed.

\subparagraph{Analyzing $\Exp[W^c_i]$.} Let $M_{i,\geq k}$ denote the
set of pseudo trapezoids from $\VD(\RR_i)$ whose conflict list has size at
least $k$. So note that $M_{i, \geq 0}$ is simply the total number of
trapezoids in $\VD(\RR_i)$.

The main idea is that the probability that a pseudo trapezoid from
$\VD(\RR_i)$ has a conflict list that is at least $t$ times the
average size $N/i$ decreases exponentially with $t$. In particular,
since every such pseudo trapezoid $\nabla$ is defined by at most
$d=O(1)$ regions from \T, and conflicts with more than $tN/i$ regions,
the probability of it appearing in $\VD(\RR_i)$ is at most
${n-d-t(N/i) \choose i} / {n \choose i} \approx
(1-(n/i))^{t(n/i)}(n/i)^d$. This is roughly $1/e^t$ times the
probability that a pseudo trapezoid with a conflict list of size $n/i$
appears in $\VD(\RR_i)$ (which, analogously, is roughly
$(1-(n/i))^{(n/i)}(n/i)^d$). This then allows us to express the
expected number of such ``heavy'' pseudo trapezoids by the expected
number of ``average'' pseudo trapezoids. More precisely, we have that
$\Exp[M_{i, \geq tN/i}] = O\left(\Exp[M_{i, \geq 0}] \cdot t^d\cdot
  \textrm{tiny}(t)\right)$, where
$\textrm{tiny}(t)=1/e^{(t/2)}$~\cite[Lemma
8.7]{harPeled2011geometric}. Har-Peled refers to this as the
exponential decay Lemma. 

We then get
\begin{align*}
  \Exp \left[\sum_{\nabla \in \VD(\RR_i)} |C^\T_\nabla|^c \right]
      &\leq \Exp\left[ \sum_{t \geq 1} \left( t \frac{N}{i} \right)^c (M_{i,\geq
        (t-1)(N/i)}-M_{i, \geq t(N/i)})  \right] \\
      &\leq \Exp\left[ \left( \frac{N}{i} \right)^c
                  \sum_{t \geq 0} (t+1)^c M_{i, \geq t(N/i)} \right] \\
      &\leq \left( \frac{N}{i} \right)^c
                  \sum_{t \geq 0} (t+1)^c ~ \Exp\left[ M_{i, \geq t(N/i)} \right]
        \text{ \{exponential decay lemma\} } \\
      &\leq \left( \frac{N}{i} \right)^c
                  \sum_{t \geq 0} (t+1)^cO\left(t^d \, \textrm{tiny}(t)
        ~\Exp\left[ M_{i, \geq 0} \right] \right)  \\
      &\leq O\left( \left( \frac{N}{i} \right)^c ~\Exp\left[ M_{i,\geq 0} \right]
                  \sum_{t \geq 0} (t+1)^{c+d} \, \textrm{tiny}(t) \right)\\
      &= O\left( \left( \frac{N}{i} \right)^c \Exp\left[ M_{i, \geq 0}
        \right] \right)
      = O\left( \left( \frac{N}{i} \right)^c i^2 \right).
\end{align*}

So in particular, for $i=r$, we obtain $\Exp[W^c_r]=O(r^2(N/r)^c)$ as
claimed. This completes the proof. \qedhere
\end{proof}

\begin{lemma}
  \label{lem:cuttings}
  Let \T be a set of $N$ Tarski cells, and $r \in \upto{1}{N}$ a
  parameter. A $1/r$-cutting $\Xi$ of \T of size $O(r^2)$, together
  with its conflict lists, can be computed in expected $O(Nr)$ time.
\end{lemma}

\begin{proof}
  The fact that $\Xi$ exists, and has size $O(r^2)$ follows by
  combining the approach of Agarwal and
  \Matousek~\cite{agarwal94range_searc_semial_sets} with the results
  of de Berg and Schwarzkopf~\cite{berg95cuttin}. In particular,
  Agarwal and \Matousek{} show that a $1/r$-cutting of size
  $O(r^2\log^2 r)$ exists (essentially the vertical decomposition of a
  random sample of size $O(r\log r)$ will likely be a
  $1/r$-cutting). Plugging this into the framework of de Berg and
  Schwartzkopf gives us that $1/r$-cutting of size $O(r^2)$ exists. We
  can use the algorithm in Lemma~\ref{lem:vertical_decomposition} to
  construct the vertical decomposition as needed by the Agarwal and
  \Matousek algorithm, as well as to implement the framework of de
  Berg and Schwartzkopf.

  For completeness, we include the full argument here.
  Let \E denote the family of Tarski cells in $\R^2$, and observe that
  $\T \subset \E$. The range space
  $(\E,\{ \C_\Delta \mid \Delta \in \E \})$ has constant
  VC-dimension~\cite{agarwal94range_searc_semial_sets}. Furthermore,
  for some subset $\T' \subset \E$ of size $m$, the vertical
  decomposition $\VD(\T')$ of the arrangement of $\T'$ has size
  $O(m^2)$. Therefore, by Lemma 3.1 of Agarwal and
  \Matousek~\cite{agarwal94range_searc_semial_sets}, the vertical
  decomposition $\VD(\NN)$ of a \emph{$1/r$-net} \NN of \T of size
  $m=O(r\log r)$ is an $1/r$-cutting of \T of size
  $O(m^2)=O(r^2\log^2 r)$. (Recall that \NN is an $1/r$ net for \T if
  for every range $\C_\Delta$ of size at least $|\T|/r$ contains
  at least one region from \NN.)

  We can compute such an $1/r$-net, together with its conflict lists,
  in expected $O(Nm)=O(Nr\log r)$ time. Indeed, taking a random sample
  \NN of size $O(r\log r)$ is expected to be a $1/r$-net for \T with
  constant probability. So, we can simply construct $\VD(\NN)$ and its
  conflict lists using Lemma~\ref{lem:vertical_decomposition}. If \NN
  is not a $1/r$-net (i.e.\ the size of the conflict lists grows
  beyond size $O(Nr\log r)$), we discard the results, and start fresh
  with a new random sample. The expected number of retries is
  constant, and thus the expected time to construct a $1/r$-cutting is
  $O(Nr\log r)$.

  We now use the results of de Berg and Schwarzkopf to obtain a
  cutting of size $O(r^2)$~\cite{berg95cuttin}. Let
  $\RR \subseteq \T$ be an arbitrary subset of regions from \T. The
  vertical decomposition has the role of what de Berg and Schwarzkopf
  call a ``canonical triangulation''. In particular, observe that we
  have the following properties:
  \begin{description}
  \item[C1] Each pseudo-trapezoid $\Delta \in \VD(\RR)$ is defined by
    a constant size defining set $\D \subseteq \RR$. More precisely,
    we need that $\Delta$ is a pseudo-trapezoid in $\VD(\D)$.

  \item[C2] For each pseudo-trapezoid $\Delta \in \VD(\RR)$, the
    pseudo-trapezoid also appears in $\VD(\T)$ if and only if the
    conflict list $\C_\Delta$ with respect to \T is empty.
  \item[C3] For a parameter $t \geq 1$, there exists a $1/t$-cutting
    of $\RR$ of size $O(t^2\log^2 t)=O(t^3)$ (as we argued above).
  \end{description}
  Therefore, we can apply the results of de Berg and
  Schwarzkopf~\cite{berg95cuttin}. In particular, their Lemma 1 gives
  us that exists an $1/r$-cutting of $\T$ of size $O(r^2)$ (since the
  expected number of pseudo-trapezoids in $\VD(\RR)$ in a random
  sample $\RR \subseteq \T$ of size $r$ is $O(r^2)$.

  All that remains is to argue that we can construct such a cutting in
  expected $O(nr)$ time. We follow a similar presentation as
  Har-Peled~\cite{harPeled2011geometric}. We take a random sample $\RR$
  of size $r$, and compute the vertical decomposition $\VD(\RR)$ and its
  conflict lists using Lemma~\ref{lem:vertical_decomposition}. For
  each pseudo-trapezoid $\nabla \in \VD(\RR)$ whose conflict list is
  larger than $CN/r$, for some constant $C$, we compute a
  $1/t_\nabla$-cutting of $|\C_\nabla|$, for
  $t_\nabla = |\C_\nabla|r/N$ and clip it to $\nabla$. We use
  the approach based on $1/t_\nabla$-nets for these cuttings. Hence,
  the expected time to construct all of them is
  \begin{align*}
    \sum_{\nabla \in \VD(\RR)} O(|\C_\nabla|t_\nabla\log t_\nabla)
    &= \sum_{\nabla \in \VD(\RR)} O(|\C_\nabla|t_\nabla^2)
     = \sum_{\nabla \in \VD(\RR)} O(|\C_\nabla|^3r^2/N^2)\\
    &= O(r^2/N^2)\sum_{\nabla \in \VD(\RR)} |\C_\nabla|^3.
  \end{align*}
  By Lemma~\ref{lem:vertical_decomposition} the total expected time is
  therefore $O((r^2/N^2) r^2 (N/r)^3 ) = O(Nr)$. Clipping the
  resulting cells can be done in the same time. By Lemma 1 of Berg and
  Schwarzkopf~\cite{berg95cuttin} the result is a $1/r$-cutting.
\end{proof}

%%%%%%%%%%%%%%%%%%%%%%%%%%%%%%%%%%%%%%%%%%%%%%%%%%%%%%%%%%%%%%%%%%%%%%%%
\section{A data structure when $\T_s$ and $\T_t$ are given}
\label{sec:A_data_structure_for_when_Rs_and_Rt_are_given}
%%%%%%%%%%%%%%%%%%%%%%%%%%%%%%%%%%%%%%%%%%%%%%%%%%%%%%%%%%%%%%%%%%%%%%%%%%%%%%%%%%%%%%%%%%%%
In this section, we consider the subproblem of finding the minimum distance over a fixed set of regions. Let $\A, \B \subseteq \T$ be two sets of regions, and let
  $m=\min\{|\A|,|\B|,n\}$. We construct a data structure on the regions in $\A$ and $\B$ such that we can compute the lower envelope $\L^{\A \times \B}(s,t)$ for any query points $s \in \bigcap \A$ and $t \in \bigcap \B$ efficiently. We call such a data structure a \LowerEnvelope data structure.

  Observe that all regions in \A overlap in some point $a$, and all
  regions in \B overlap in some point $b$. So $\A \subseteq \T_a$, and
  $\B \subseteq \T_b$. It then follows from Lemma~\ref{lem:size_rels}
  that $\Rel(\A \times \B)$ has size $O(m)=O(\min
  \{|\A|,|\B|,n\})$. The functions $f_{ST}$, for
  $(S,T) \in \Rel(\A \times \B)$ are four-variate, and have constant
  algebraic degree. 
  The \LowerEnvelope data structure thus has to be constructed on only $O(m)$ four-variate functions. Next, we consider two \LowerEnvelope data structures that can be built on a given set of functions $f_{ST}$.

\begin{lemma}
  \label{lem:visible_vertices}
  For any constant $\eps > 0$, we can construct a \LowerEnvelope data structure of size $O(m^{6+\eps})$ in $O(m^{6+\eps})$ expected time that can answer queries in $O(\log m)$ time.
\end{lemma}

\begin{proof}
One way to obtain efficient queries of $\L(s,t)$ is by
  storing the vertical decomposition of (the graphs of) the functions $f_{ST}$
  for all $(S,T) \in \Rel(\A \times \B)$. 
  Koltun~\cite{koltun04almos}
  shows that the vertical decomposition of $m$ surfaces in $\R^d$, each
  described by an algebraic function of constant degree, has
  complexity $O(m^{2d-4+\eps})$, and can be stored in a data structure
  of size $O(m^{2d-4+\eps})$. This data structure can also be constructed in $O(m^{2d-4+\eps})$ expected time, and we can query the value of the
  lower envelope by a point location query in $O(\log m)$ time~\cite{chazelle91singl_expon_strat_schem_real}.

  Since our functions are
  four-variate, we have $d=5$, and thus we get an $O(m^{6+\eps})$
  size data structure that answers queries in $O(\log m)$ time.
\end{proof}

\begin{lemma}\label{lem:vertical_ray_shooting}
    For any constant $\eps > 0$, we can construct a \LowerEnvelope data structure of size $O(m^{5+\eps})$ in $O(m^{5+\eps})$ expected time that can answer queries in $O(\log^2 m)$ time.
\end{lemma}
\begin{proof}
An alternative way to query $\L(s,t)$ is to perform a vertical ray-shooting query. Agarwal~\etal~\cite{Agarwal21_polynomial_partitioning} show that a collection of $m$ semialgebraic sets in $\R^d$, each of constant complexity, can be stored in a data structure of size $O(m^{d + \eps})$, for any constant $\eps > 0$, that allows for vertical ray-shooting queries in $O(\log^2 m)$ time. The data structure can also be constructed in $O(m^{d + \eps})$ expected time. As the (graphs of) our functions $f_{ST}$ are semialgebraic sets in $\R^5$, this gives an $O(m^{5 + \eps})$ size data structure that can answer queries in $O(\log ^2 m)$ time.
\end{proof}

%%%%%%%%%%%%%%%%%%%%%%%%%%%%%%%%%%%%%%%%%%%%%%%%%%%%
\section{A data structure when $\T_s$ is given}
\label{sec:given_Rs}
%%%%%%%%%%%%%%%%%%%%%%%%%%%%%%%%%%%%%%%%%%%%%%%%%%%%
\newcommand{\zlarge}{Z}
\newcommand{\zsmall}{r}

Let $\A, \B \subseteq \T$ be two subsets of regions with $|\A|=k$,
and $|\B|=M_0$. We develop a data structure
to store $\L^{\A \times \B}$ that can efficiently compute
$\L^{\A \times \B}(s,t)$, provided that $s \in \bigcap \A$. As
before, this implies that $\T_s \supseteq \A$, and thus
$k=O(n)$. Moreover, if $\A=\T_s$ and $\B=\T$ this thus allows us to
compute $d(s,t)$ for any $t \in P$. As there are at most $nk$ relevant regions for $t$, we can assume that $M_0 \leq nk$. We can compute these relevant regions, and store for each region in $\A$ a bidirectional pointer to its relevant region in $\B$, in $O(n^2\log n)$ time by sorting the regions on their respective shortest path map and performing a linear search. We formulate our result with respect to any \LowerEnvelope data structure that uses $\storage(m)$ space and expected preprocessing time on $m$ functions, and has query time $\query(m)$, where $\storage(m)$ is of the form $Cm^a$ for some constants $C >0$ and $a \geq 3$, and $\query(m)$ is non-decreasing.
In this section, we prove the following lemma.

\begin{lemma}
  \label{lem:subproblem_ds}
  For any constant $\eps > 0$, there is a data structure of size
  $O(n^{2+\eps} \storage(k))$, so that for any query points $s,t$ for which
  $s \in \bigcap \A$ we can compute $\L^{\A \times \B}(s,t)$ in
  $O(\log n + \query(k))$ time. Building the data structure takes
  $O(n^{2+\eps} \storage(k))$ expected time.
\end{lemma}

Note that if we are somehow given $\A=\T_s$, and have $\B=\T$ then
$k=n$, and thus we get an $O(n^{2+\eps}\storage(n))$ size data
structure that can compute $d(s,t)$ in $O(\log n + \query(n))$ time.

Our data structure is essentially a cutting-tree~\cite{Clarkson87_new_applications_random_sampling} in which
each node stores an associated \LowerEnvelope data
structure (Section~\ref{sec:A_data_structure_for_when_Rs_and_Rt_are_given}). In detail, let $r \in \upto{2}{M_0}$ be a parameter to be determined later. We build a
$1/r$-cutting $\Xi'$ of $\B$ using the algorithm from
Lemma~\ref{lem:cuttings}. 
Furthermore, we preprocess $\Xi'$ for point
location
queries~\cite{edelsbrunner86optim_point_locat_monot_subdiv}. For each
cell $\Delta \in \Xi'$, the subset $\B_\Delta$ of regions of $\B$
that contain $\Delta$ have an apex visible from any point within
$\Delta$. (Since all points in a region $T$ see its
apex $v_T$, i.e.\ the line segment to $v_T$ does not intersect $\delta P$, and $\Delta$ is contained in $T$.) Hence, we store a
\LowerEnvelope data structure on the pair of sets
$(\A,\B_\Delta)$ for each cell $\Delta$.
We now recursively process the set $\C_\Delta$ of
 regions whose boundary intersects $\Delta$. 

All that remains is to describe how to choose the parameter $r$ that we use to build the cuttings. Let $c$ be a constant so
that the number of cells in a $1/r$-cutting on $\B$ has at most
$cr^2$ cells. The idea is to pick $r=\max\{n^\delta,2,2c^{1/(a-2)}\}$, for some fixed $\delta \in (0,1)$ to be
specified later.

\subparagraph{Space usage.} Let $M$ be the number of remaining regions from $\B$
(initially $M=M_0$).
Storing the cutting $\Xit$, and its point location structure takes
$O(r^2)$
space~\cite{edelsbrunner86optim_point_locat_monot_subdiv}. Moreover,
for each of the $O(r^2)$ cells, we store a \LowerEnvelope data
structure of size $\storage(\min(k,M))$. There are only $M/r$ regions whose boundary intersects a cell
of the cutting on which we recurse. The space usage of the data structure is thus given by the following recurrence:
\[ \S(M)=
  \begin{cases}
    cr^2\S(M/r) + O\left(r^2 \cdot \storage(\min\{k,M\})\right) & \text{if } M > 1\\
    O(1)                            & \text{if } M = 1.
  \end{cases}
\]
After $i$ levels of recursion, this gives $O(r^{2i}) = O(n^{2\delta i})$ subproblems of size 
at most $M_0/r^i \leq nk/n^{\delta i} = n^{1-\delta i} k$. It follows that at level $\frac{1}{\delta}$, we have $O(n^2)$ subproblems of size $O(k)$. 
Next, we analyze the space usage by the higher levels, where $i \leq \frac{1}{\delta}$, and lower levels, where $i > \frac{1}{\delta}$, separately. 
Note that for the higher levels, we have $\min\{k,M\} = k$, and for the lower levels $\min\{k,M\} = M$. There are only $\frac{1}{\delta} = O(1)$ higher levels.

In the higher levels of our data structure, the $O(r^{2i})$ associated \LowerEnvelope data structures at level $i \leq \frac{1}{\delta}$, take $O(r^{2i} \cdot r^2 \cdot \storage(k)) = O(n^{2 + \eps}\storage(k))$ space, by setting $\delta = \eps/2$. Since there are only $O(1)$ higher levels, the total space usage of these levels is $O(n^{2 + \eps}\storage(k))$ as well.

In the lower levels of our data structure, we are left with $O(n^2)$
``small'' subproblems at level $i = \frac{1}{\delta}$. In the
following we argue that such a ``small'' subproblem uses only
$O(r^2 \storage(k))$ space, and therefore the lower levels also use
only $O(n^{2 + \eps}\storage(k))$ space in total.

For each such ``small'' subproblem, we have a
$1/r$-cutting on $M \leq k$ regions that consists of at most $cr^2$
cells.  Since $M \leq k$, the recurrence simplifies to
\[ \S'(M)=
  \begin{cases}
    cr^2\S'(M/r) + dr^2\storage(M) & \text{if } M > 1\\
    e                            & \text{if } M = 1,
  \end{cases}
\]
where $c,d$, and $e$ are positive constants.

\begin{lemma}
  \label{lem:recurrence}
  Let $c,d,e,\eps > 0$ be constants and $r > \max\{2,2c^{1/(a-2)}\}$. The
  recurrence $\S'(M)$
  solves to $O(r^2\storage(M))$.
\end{lemma}

% \begin{proof}
%   We prove by induction on $M$ that $\S(M) \leq Dr^2M^{6+\eps}$,
%   for constant $D \geq \max\{2d,e\}$. The base case $M=1$ is
%   trivial. Using the induction hypothesis we then have that
%   \begin{align*}
%     cr^2\S(M/r) + dr^2\storage(M)
%     \leq
%       cr^2\left(Dr^2\storage\left(\frac{M}{r}\right)\right) + dr^2\storage(M)
%     =
%       cDr^4\storage\left(\frac{M}{r}\right) + dr^2\storage(M)
%   \end{align*}

%   Using some basic calculus we then find that:
%   \begin{align*}
%     cDr^4\storage\left(\frac{M}{r}\right) + dr^2\storage(M) &=
%     cDr^4 \frac{M^a}{r^a} + dr^2M^a \\
%     &= r^2M^a\left(\frac{cD}{r^{a-2}}+d\right) \\
%     &\leq r^2M^a\left(\frac{cD}{(2c^{1/(a-2)})^{a-2}}+d\right)\\
%     &\leq r^2M^a\left(\frac{D}{2}+d\right)\\
%     &\leq Dr^2M^a = Dr^2\storage(M)
%   \end{align*}
% Here, we used that $a \geq 3$, $r > 2c^{1/(a-2)}$, and $D \geq 2d$.
%   This completes the proof.
% \end{proof}

\begin{proof}
  We prove by induction on $M$ that $\S'(M) \leq Dr^2\storage(M)$,
  for constant $D \geq \max\{2d,e\}$. The base case $M=1$ is
  trivial, as $D \geq e$. Recall that $\storage(m)$ is of the form $Cm^a$ for some constants $C >0$ and $a \geq 3$. Using the induction hypothesis we then have that
  \begin{align*}
    cr^2\S'(M/r) + dr^2\storage(M)
    \leq
      cr^2\left(Dr^2\storage\left(\frac{M}{r}\right)\right) + dr^2\storage(M)
    =
      cDr^4\storage\left(\frac{M}{r}\right) + dr^2\storage(M).
  \end{align*}

  Using some basic calculus we then find that:
  \begin{align*}
    cDr^4\storage\left(\frac{M}{r}\right) + dr^2\storage(M) &=
    cDr^4C \frac{M^a}{r^a} + dr^2CM^a \\
    &= Cr^2M^a\left(\frac{cD}{r^{a-2}}+d\right) \\
    &\leq Cr^2M^a\left(\frac{cD}{(2c^{1/(a-2)})^{a-2}}+d\right)\\
    &\leq Cr^2M^a\left(\frac{D}{2}+d\right)\\
    &\leq CDr^2M^a = Dr^2\storage(M)
  \end{align*}
Here, we used that $a \geq 3$, $r > 2c^{1/(a-2)}$, and $D \geq 2d$.
  This completes the proof.
\end{proof}

\subparagraph{Analyzing the preprocessing time.}  By
Lemma~\ref{lem:cuttings} we can construct a $1/r$-cutting on $\B$,
together with the conflict lists, in $O(rM)$ expected time. Since the cutting
has size $O(r^2)$, we can also preprocess it in $O(r^2)$ time for
point location
queries~\cite{edelsbrunner86optim_point_locat_monot_subdiv}. Note that we can compute the set of relevant pairs in a subproblem in $O(\min\{k,M\})$ time, as we already computed pointers between each such pair. The expected time
to build each associated structure is thus $\storage(\min\{k,M\})$. For the small
subproblems, the $O(rM)$ term is dominated by the $O(r^2\storage(M))$
time to construct the associated data structures, and thus we obtain
the same recurrence as in the space analysis (albeit with different
constants).

At each of the $O(1)$ higher levels, we spend
$O(r^{2i} \cdot rM) = O(r^{2i + 1} M_0/r^i) = O(n^{2\delta i
  + \delta} nk) = O(n^{2 + \delta}k)$ expected time to construct the cutting,
and $O(r^{2i} r^2 \storage(k)) = O(n^{2+2\delta} \storage(k))$ expected time to
construct the associated data structures. Since the
$O(n^{2 + 2\delta} \storage(k))$ term
dominates, we thus obtain an expected preprocessing time of
$O(n^{2 + 2\delta} \storage(k)) = O(n^{2+\eps}\storage(k))$.

\subparagraph{Querying.} Let $s,t$ be the query points. We use the
point location structure on $\Xi'$ to find the cell $\Delta$ that
contains $t$, and query its associated
\LowerEnvelope data structure to compute
$\L^{\A \times \B_\Delta}(s,t)$. We recursively query the structure
for on the regions whose boundaries intersect $\Delta$. When we have
only $O(1)$ regions left in $\B_\Delta$, we compute
$\L^{\A \times \B_\Delta}(s,t)$ by explicitly evaluating $f_{ST}(s,t)$
for all $O(1)$ pairs of relevant regions (one per region in
$\B_\Delta$), and return the minimum
found. A region that contains $t \in \Delta$ either contains a cell
$\Delta$, or intersects it. Moreover, all regions that contain
$\Delta$ contain $t$. Hence, the sets $\B_\Delta$ over all cells
$\Delta$ considered by the query together form a partition of
$\B_t$. Since we compute $\L^{\A \times \B_\Delta}(s,t)$ for each
such set, it then follows by Lemma~\ref{lem:decompose} that our
algorithm correctly computes $\L^{\A \times \B_t}(s,t)$, provided
that $s \in \bigcap \A$.

Finding the cell $\Delta$ containing $t$ takes $O(\log r)$ time,
whereas querying the \LowerEnvelope data structure
to compute $\L^{\A \times \B_\Delta}(s,t)$ takes
$\query(\min\{k,M\})$ time. Hence, we spend
$O(\log r + \query(k)) = O(\log n + \query(k))$ time at every level of the recursion. Observe
that there are only $O(2/\delta) = O(1)$ levels in the recursion, as $M_0/r^{(2/\delta)} \leq n^2/n^{(\delta \cdot 2/\delta)} = 1$. It follows that the total query time is $O(\log n + Q(k))$ as claimed.

\section{A two-point shortest path data structure}
\label{sec:A_two-point_shortest_path_data_structure}

Let $\T$ be the set of all $N=O(n^2)$ SPM regions, and let
$\A \subseteq \T$ be a subset of $K$ such regions. We develop a data structure to store $\L^{\A \times \T}$, so
that given a pair of query points $s,t$ we can query
$\L^{\A \times \T}(s,t)$ efficiently. Our main idea is to use
a similar approach as in Section~\ref{sec:given_Rs}; i.e.\ we build a
cutting tree that allows us to obtain $\A_s$ as $O(1)$-canonical
subsets, and for each such set $\A'$ we query a
Lemma~\ref{lem:subproblem_ds} data structure to compute
$\L^{\A' \times \T}(s,t)$. When $\A=\T$ using the Lemma~\ref{lem:visible_vertices} \LowerEnvelope gives us an
$O(n^{10+\eps})$ space data structure that can be queried for
$\L^{\T \times \T}(s,t) = d(s,t)$ in $O(\log n)$ time. Alternatively, using the Lemma~\ref{lem:vertical_ray_shooting} \LowerEnvelope data structure, we obtain an $O(n^{9 +\eps})$ space data structure that can answer queries in $O(\log^2 n)$ time.

Let $r \in \upto{2}{N}$ be a parameter. As before, let  $\storage(m)$ denote the space and expected preprocessing time, and $\query(m)$ the query time, of a \LowerEnvelope data structure on $m$ functions, where $\storage(m) = m^a$ for some constant $a \geq 3$ and $\query(m)$ is non-decreasing. We build a
$1/r$-cutting $\Xi$ of $\A$, and preprocess it for point location
queries~\cite{edelsbrunner86optim_point_locat_monot_subdiv}. For each
cell $\Delta \in \Xi$, we consider the set of $k \leq \min\{K,n\}$ regions $\A_\Delta$
fully containing $\Delta$, and build the data structure from
Lemma~\ref{lem:subproblem_ds} on $(\A_\Delta,\T)$ that can be
efficiently queried for $\L^{\A_\Delta \times \T}(s,t)$ when
$s \in \Delta$. This structure takes
$O(n^{2+\eps'}\storage(k))$ space,
for an arbitrarily small $\eps' > 0$, and achieves
$O(\log n + \query(k))$ query time. We recursively process all
regions from $\A$ whose boundaries intersect $\Delta$ (i.e.\ the at
most $K/r$ regions in the conflict list of $\Delta$). Since the
cutting consists of at most $cr^2$ cells, for some constant $c$, the
space usage over all cells of $\Xi$ is $O(r^2n^{2+\eps'}\storage(k))$.

\subparagraph{Space analysis.} We again pick
$r=\max\{n^\delta,2,2c^{1/(a-2)}\}$, for some arbitrarily small $\delta > 0$.

In particular we will choose
$\delta=\eps/4$ and $\eps'=\eps/2$, so that $\eps'+2\delta=\eps$. The
space usage of the data structure then follows the recurrence
\begin{equation*}\label{eq:recursion_s}
      \S_2(K) =
  \begin{cases} cr^2\S_2(K/r) + O\left(r^2n^{2+\eps'}\storage(\min\{K,n\})\right) & \text{if }K >
      1\\
      O(1)                            & \text{if } K = 1.
  \end{cases}
\end{equation*}

As we start with $K = N$ regions, we have $O(r^{2i})=O(n^{2\delta i})$
subproblems after $i$ levels of recursion, each of size at most
$N/r^i=N/n^{\delta i} = O(n^{2-\delta i})$. It follows that after
$\frac{1}{\delta} = O(1)$ levels, we thus have at most
$c'n^{2\delta\frac{1}{\delta}}=c'n^2$, where $c'=c^{\frac{1}{2\delta}}$ is some constant,
``small'' subproblems, each of size $O(n^{2-\delta\frac{1}{\delta}})=O(n)$.

We bound the space used by the higher levels (when $K>n$) and the
lower levels of the recursion separately.
 We start with the higher
levels.
At level $i$, the $O(n^{2\delta i})$ subproblems contribute a total of
$\displaystyle O\left(n^{2\delta i}n^{2\delta}n^{2+\eps'} \cdot \storage(n)\right) = O\left(n^{4+2\delta+\eps'}\storage(n)\right)$
space for $i \leq \frac{1}{\delta}$. Using that $\eps'+2\delta=\eps$, this is thus a total of $O(n^{4+\eps}\storage(n))$
space for level $i \leq \frac{1}{\delta}$. Since we only have $O(1)$ higher levels, their total
space usage is $O(n^{4+\eps}\storage(n))$ as well.

Once we are in the ``small'' case, $K \leq n$, we have $\min\{K,n\} = K$, and we remain in the
``small'' case. So, for $K \leq n$ we actually have the recurrence

 \begin{equation*}\label{eq:recursion_s_small}
 \S'_2(K)=
  \begin{cases}
    cr^2\S'_2(K/r) + dr^2n^{2+\eps'}\storage(K) & \text{if } K > 1\\
    e                             & \text{if } K = 1.
  \end{cases}
\end{equation*}
Which, using a similar analysis as in Lemma~\ref{lem:recurrence} solves
to $O(r^2n^{2+\eps'}\storage(K))$. Using that
$r = \max\{n^{\delta}, 2, 2c^{1/(a-2)}\}$, and that the maximum size of a
``small'' subproblem is only $n$, each such subproblem thus uses only
$O(n^{2+\eps}\storage(n)) $ space. Since we have $O(n^2)$ small subproblems
the total space used for the ``small'' subproblems is $O(n^{4+\eps}\storage(n))$
space.

Hence, it follows the structure uses a total of $O(n^{4+\eps}\storage(n))$ space.

\subparagraph{Preprocessing time analysis.} As in
Section~\ref{sec:given_Rs} the preprocessing time for the ``small''
subproblems follows the same recurrence as the space bound. Hence,
constructing the data structure for each such subproblem takes
$O(n^{2+\eps}\storage(n))$ expected time, and thus $O(n^{4+\eps}\storage(n))$ expected time in total. For the
higher levels: constructing a $1/r$-cutting on a subproblem at
level $i$ takes $O(r(N/r^i))$ expected time. Since there are $O(r^{2i})$
subproblems on level $i$, and we have $i \leq \frac{1}{\delta}$, it thus
follows that we spend $O(r^{2i} \cdot r(N/r^i)) = O(r^{i + 1}N) = O(n^{\delta i + \delta}N) = O(n^{3+\eps})$ expected time per level
to build the cuttings. As in the space analysis, the expected time to
construct the associated data structures of level $i$ is
$O(n^{2+2\delta}n^{2+\eps'}\storage(n)) = O(n^{4+\eps}\storage(n))$, which dominates the time to construct the cuttings. Since the number of higher
levels is constant, it follows the total expected preprocessing time is also
$O(n^{4+\eps}\storage(n))$.

\subparagraph{Query analysis.} 
We query the data structure symmetrically to Section~\ref{sec:given_Rs}, using the point location structure on $\Xi$ to find the cell $\Delta$ that contains $s$. We then query its associated Lemma~\ref{lem:subproblem_ds} data structure to compute $\L^{\A_\Delta \times \T}(s,t)$, and recurse to compute $\L^{(\A_s \setminus \A_\Delta) \times \T}(s,t)$. Note that $\A_s \setminus \A_\Delta$ is indeed a subset of the regions on which we recursively built a cutting, as $s \in \Delta$ and the region boundary intersects $\Delta$. Finally, we return the smallest distance found.

In total there are only $O(1)$ levels,
at each of which we spend $O(\log r) = O(\log n)$ time to find the cell that contains $s$, and then $O(\log n + \query(n))$ time to
query the Lemma~\ref{lem:subproblem_ds} structure. The resulting query time is thus $O(\log n + \query(n))$.

From the \LowerEnvelope data structure, we do not obtain just the length of the path, it also provides us with a vertex $v$ on the shortest path. To compute the actual path, we locate both $s$ and $t$ in the shortest path map of $v$, and then follow the apex pointers to return the path. This takes only $O(\log n + k)$ additional time, where $k$ denotes the length of the path.

\subparagraph{Putting everything together.} As we remarked earlier, we
can construct the set of regions \T in $O(n^2\log n)$ time, and store
them using $O(n^2)$ space. Similarly, constructing a data structure to
test if $s$ and $t$ can directly see each other also takes
$O(n^2\log n)$ time $O(n^2)$ space. We thus obtain the following lemma:
\begin{lemma}
  \label{lem:2pt_shortest_path_query_data_structure_general}
  For any constant $\eps >0$, we can build a data structure using $O(n^{4+\eps}\storage(n))$
  space and expected preprocessing time that can answer two-point shortest path queries in $O(\log n + \query(n))$ time.
\end{lemma}

By applying this lemma to the \LowerEnvelope data structures of Lemma~\ref{lem:visible_vertices} and~\ref{lem:vertical_ray_shooting} we obtain our main result:
\main*

\subparagraph{Remark.} Note that it is also possible to avoid using a multi-level data structure and use only a single cutting tree for both $s$ and $t$.
In that case, we would build a \LowerEnvelope data structure for every pair of nodes of the cutting tree.
However, we focus on the multi-level data structure here, as we do need this structure for the case where $s$ is restricted to the boundary of $P$.

\section{Space-time trade-off}
\label{sec:Space-_time_trade-off}

We can achieve a trade-off between the space usage and the query time by grouping the polygon vertices. We group the vertices into $\ell$ groups $V_1,\dots,V_\ell$ of size $O(n/\ell)$. We still compute the multi-set of regions $\T$ as before, but then partition the regions into $\ell$ multi-sets $\T_1,\dots,\T_\ell$ where $\T_i$ contains all regions generated by vertices in $V_i$. Note that each of these sets contains $O(n^2/\ell)$ regions. For each of these sets, we build the data structure of  Lemma~\ref{lem:2pt_shortest_path_query_data_structure_general}.
Each data structure is built on $O(n^2/\ell)$ regions, of which there are only $O(n/\ell)$ relevant pairs. In the notation of Section~\ref{sec:A_two-point_shortest_path_data_structure}, we have $N = O(n^2/\ell)$ regions and the space of the data structure of Lemma~\ref{lem:subproblem_ds} is $O(n^{2+\eps'} \storage(\min\{K, n/\ell\}))$. Next, we analyze the space usage for the same parameter choice for $r$ as in Section~\ref{sec:A_two-point_shortest_path_data_structure}.
We consider a subproblem ``small'' whenever $K \leq n/\ell$. 

The analysis of the space usage is similar to the approach in Section~\ref{sec:A_two-point_shortest_path_data_structure}. At level $i$ of the recursion, we have $c^i n^{2\delta i}$ subproblems of size $N/r^i = O(n^{2-\delta i}/\ell)$. After $i = \frac{1}{\delta} = O(1)$ levels we are left with $c^in^{2\delta i}= O(n^2)$ subproblems of size $O(n^{2-\delta i}/\ell) = O(n/\ell)$. The $O(1)$ large levels have a space usage of 
$O(n^{2\delta i} \cdot n^{2\delta} \cdot n^{2+\eps'} \cdot \storage(n/\ell)) = O(n^{4+\eps}\storage(n/\ell))$. For the small subproblems, i.e.\ $K \leq n/\ell$, we have the recurrence $\S_2(K) = cr^{2}\S_2(K/r)+O(r^2 \cdot n^{2+\eps'} \storage(K))$. This again solves to $\S_2(K) = O(r^2 n^{2+\eps'} K^{6+\eps'}) = O(n^{2+\eps}\storage(n/\ell))$. The $O(n^2)$ small subproblems thus use $O(n^{4+\eps}\storage(n/\ell))$ space in total. The total space usage of all $\ell$ data structures is thus $O(\ell n^{4+\eps}\storage(n/\ell))$.

To answer a query, we simply query each of the $\ell$ data structures for a vertex on the shortest path between $s$ and $t$. We then compute the length of these paths in $P$ and return the shortest path. Queries can thus be answered in $O(\ell (\log n + Q(n/\ell))$ time. The following lemma summarizes the result.
\begin{lemma}\label{lem:trade_off_general}
    For any constant $\eps >0$ and integer $1 \leq \ell \leq n$, we can build a data structure using $O(\ell n^{4+\eps}\storage(n/\ell))$
  space and expected preprocessing time that can answer two-point shortest path queries in $O(\ell(\log n + \query (n/\ell)))$ time.
\end{lemma}

By applying the lemma to the \LowerEnvelope data structure of Lemma~\ref{lem:vertical_ray_shooting}, we obtain the following result.
\tradeoff*

\section{A data structure for $s$ and/or $t$ on the boundary}
\label{sec:A_data_structure_for_$s$_on_the_boundary}

In this section, we show that we can apply our technique to develop data structures for the restricted case where either one or both of the query points must lie on the boundary of $P$. In case only $s$ is restricted to the boundary, but $t$ can be anywhere in the interior of $P$, we obtain an $O(n^{6+\eps})$ space data structure. In case both $s$ and $t$ are restricted to the boundary, the space usage decreases further to $O(n^{4+\eps})$. This improves a result of Bae and Okamoto~\cite{bae12query}, who presented a data structure using roughly $O(n^{5+\eps})$ space for this problem.

The main idea is to use the same overall approach: we subdivide the
space into (possibly overlapping) regions, and construct a data structure that can report the regions stabbed by $s$ as a few canonical subsets. 
For each such subset, we build a data structure to
find the regions stabbed by $t$ and report them as canonical subsets,
and for each of those canonical subsets, we store the lower envelope of
the distance functions so that we can efficiently query $d(s,t)$.

The two differences are that: (i) now the set of candidate regions for
$s$ and $t$ might differ; for $s$ this is now the set $\I = \{T \cap
\partial P \mid T \in \T\}$ of intervals formed
by the intersection of the \SPM regions with the boundary of $P$; for
$t$ they are the set \T or \I, depending on whether $t$ is restricted to the boundary or not, and (ii) the
functions $f_{ST}(s,t)$ are no longer four-variate. When only $s$ is restricted to $\partial P$ these functions are three-variate, and when both $s$ and $t$ are restricted to $\partial P$ they are even bivariate. The following two lemmas describe the data structures for when the set of regions stabbed by $s$ is given (i.e.\ the data structure from
Section~\ref{sec:given_Rs}) for these two special cases.

\begin{lemma}
  \label{lem:restricted_subproblem_ds}
  Let $\A \subseteq \I$ be a subset of $k$ intervals, and
  $\B \subseteq \T$ a set of $M_0$ regions. For any $\eps > 0$, there is a data
  structure of size $O(n^{2+\eps} k^{3+\eps})$, so that for any query
  points $s,t$ for which $s \in \bigcap \A$ we can compute
  $\L^{\A \times \B}(s,t)$ in $O(\log^2 n)$ time. Building
  the data structure takes $O(n^{2+\eps} 
 k^{3+\eps})$ expected time.
\end{lemma}
\begin{proof}
The lower
envelope of $m$ partial three-variate functions has complexity only
$O(m^{3+\eps})$, and can actually be stored using $O(m^{3+\eps})$ space
while supporting $O(\log^2 m)$ time
queries~\cite{agarwal97comput_envel_four_dimen_applic}. Applying Lemma~\ref{lem:subproblem_ds} to this \LowerEnvelope data structure proves the lemma.
\end{proof}

\begin{lemma}\label{lem:boundary_t}
    Let $\A \subseteq \I$ be a subset of $k$ intervals, and
  $\B \subseteq \I$ a set of $M_0$ intervals. For any $\eps > 0$, there is a data
  structure of size $O(n^{1+\eps} k^{2+\eps})$, so that for any query
  points $s,t$ for which $s \in \bigcap \A$ we can compute
  $\L^{\A \times \B}(s,t)$ in $O(\log n)$ time. Building
  the data structure takes $O(n^{1+\eps} k^{2+\eps})$ time.
\end{lemma}
\begin{proof}
The set $\B$ is now a set of intervals instead of a set of regions. Therefore, rather than using $1/r$-cuttings as in
Section~\ref{sec:given_Rs}, we  use
a segment tree $\Tr_1$ with fanout $f=n^\delta$, for some
$\delta \in (0,1/2)$~\cite{Comp_geom_book}. Each leaf node $\mu$ corresponds to
some atomic interval $J_\mu$. For each internal node $\nu$ we define
$J_\nu$ as the union of the $J_\mu$ intervals of its children
$\mu$. Furthermore, for every node $\mu$, let $\I_\mu$ denote the set
of intervals from \I that span $J_\mu$, but do not span the interval
$J_\nu$ of some ancestor $\nu$ of $\mu$. Let $m_\mu$ denote the
number of intervals in $\I_\mu$.

For each node $\mu$ of the tree, we store the lower envelope of the
functions $f_{ST}(s,t)$ with $S \in \A$ and $T \in \I_\mu$. As both
$\A$ and $\B$ are sets of intervals, the functions $f_{ST}(s,t)$ are
in this case only bivariate. For any $\eps'> 0$, we can store the
lower envelope in a data structure that supports $O(\log n)$ time
point location queries using $O(\min\{k,m_\mu\}^{2+\eps'})$
space~\cite{Sharir94lower_envelopes}. The data structure can be build
in $O(\min\{k,m_\mu\}^{2+\eps'})$ time~\cite{agarwal96overl_lower_envel_its_applic}.

Observe that the total size of all
\emph{canonical subsets} $\I_\mu$ is
$\sum_{\mu \in \Tr_1} m_\mu = O(n^{2+\delta})$ (as there are only
$O(2/\delta)=O(1)$ levels in the tree, at each level an interval
$I \in \I$ is stored at most $2f = 2n^\delta$ times). It then
follows that the space used by the data structure is
\begin{equation*}
        \sum_{\mu \in \Tr_1} O\left(\min\{k,m_\mu\}^{2+\eps'}\right) = O\left(k^{1+\eps'}\sum_{\mu \in \Tr_1} m_\mu\right) = O\left(k^{2+\eps'}n^{1+\delta}\right).
\end{equation*}
So, choosing $\eps'= \eps$ and $\delta = \eps$, we achieve a data structure of size $O(n^{1+\eps}k^{2+\eps})$. 
To query the data structure with a pair of
points $s \in \partial P$ and $t \in \partial P$, we find the set of intervals
stabbed by $t$, and report them as a set of $O(1)$ canonical
subsets. In particular, for each node $\mu$ on the search path, which
has height $O(1)$, we query its
associated \LowerEnvelope data structure in $O(\log n)$ time and use $O(\log f)=O(\log n)$ time to find
the right child where to continue the search.
\end{proof}

Next, we describe the two-point shortest path data structure with $s$ restricted to $\partial P$ that uses either of these two lemmas as a substructure. The set \I of candidate regions for $s$ is now simply a set of
$N=O(n^2)$ intervals. Therefore, as in Lemma~\ref{lem:boundary_t}, rather than using $1/r$-cuttings as in
Section~\ref{sec:A_two-point_shortest_path_data_structure}, we  use
a segment tree \Tr with fanout $f=n^\delta$, for some
$\delta \in (0,1/2)$~\cite{Comp_geom_book}. For every node $\mu$, let $\I_\mu$ again denote the set
of intervals that is stored at $\mu$, and let $k_\mu = |\I_\mu|$. For each such a set $\I_\mu$,
we build the data structure of either
Lemma~\ref{lem:restricted_subproblem_ds} or Lemma~\ref{lem:boundary_t}. 
Because each interval is stored at most $2f = 2n^\delta$ times at a level of the tree, and there are only $O(2/\delta) = O(1)$ levels, we have that $\sum_{\mu \in \Tr} k_\mu = O(n^{2+\delta})$.
It then
follows that the space used by the data structure for only $s$ on the boundary (using Lemma~\ref{lem:restricted_subproblem_ds}) is
\begin{equation*}\label{eq:space_boundary}
\begin{split}
   \sum_{\mu \in \Tr} O\left(n^{2+\eps'} \min\{k_\mu,n\}^{3+\eps'}\right)
  &= O \left( n^{2+\eps'}  \sum_{\mu \in \Tr} n^{2 + \eps'} k_\mu      \right)\\ 
  &= O\left( n^{4+2\eps'} \sum_{\mu \in \Tr}k_\mu \right) = O(n^{4+2\eps'}n^{2+\delta}).
  \end{split}
\end{equation*}
And the space used by the data structure for the both query points on the boundary (using Lemma~\ref{lem:boundary_t}) is
\begin{equation*}
    \begin{split}
        \sum_{\mu \in \Tr} O\left(n^{1+\eps'}\min\{k_\mu, n\}^{2+\eps'}\right) &= O\left(n^{1+\eps'}\sum_{\mu \in \Tr} n^{1+\eps'}k_\mu \right) \\
        &= O\left(n^{2+2\eps'}\sum_{\mu \in \Tr}k_\mu \right) = O\left(n^{2+2\eps'}n^{2+\delta}\right).
    \end{split}
\end{equation*}
So, picking $\eps'=\eps/4$ and $\delta = \eps/2$ we obtain a data structure of size $O(n^{6+\eps})$ or $O(n^{4+\eps})$, respectively.

\subparagraph{Querying.} To query the data structure with a pair of
points $s \in \partial P$ and $t \in P$ or $t \in \partial P$, we find the set of intervals
stabbed by $s$, and report them as a set of $O(1)$ canonical
subsets. In particular, for each node $\mu$ on the search path, which
has height $O(1)$, we query its
associated data structure and use $O(\log f)=O(\log n)$ time to find
the right child where to continue the search. We therefore obtain the
following results:

\boundary*

\boundaryst*

\section{Lower Bounds}
\label{sec:lower_bounds}

In this section we show a lower bound for the two point shortest path problem.
To do this, one would typically reduce from a problem, for which a lower bound is already known or conjectured.
In our case, we reduce from Hopcroft's problem~\cite{Jeff2000} for which lower bounds are known.
However, not the exact type of lower bound that we need.
We start by defining Hopcroft's problem.

\begin{figure}[tbp]
    \centering
    \includegraphics{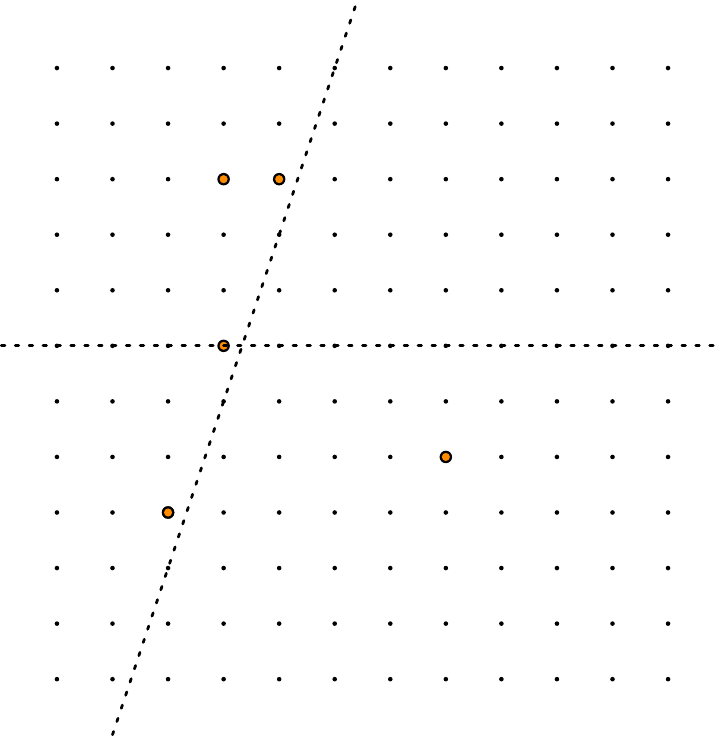}
    \caption{A set of points and a query line. We want to preprocess the points such that we can quickly answer whether the line contains any of the points.}
    \label{fig:line-emptiness}
\end{figure}

In \textit{Hopcroft's problem} we are given a set $R$ of $n$ points in the plane and a set  $\L$ of $n$ lines. 
We then ask whether there exist a line $\ell \in \L$ containing a point from $R$.
See \Cref{fig:line-emptiness} for an illustration.
We further define the integer version of Hopcroft's problem, where each point in $R$ has integer coordinates bounded by $L$. In the same spirit, we ask for every line $\ell = \{(x,y)\in\R^2 : ax+b = y \}$ that the slope and intercept ($a,b$) are integers as well with absolute value bounded by $L$.

\begin{conjecture*}[Integer Hopcroft Lower Bound]
\label{con:HopcroftLowerBound}
    Let $L = n^{O(1)}$ and consider a data structure to solve the integer Hopcroft problem with space complexity $\spacecomplexity$ and query time $\querytime$.
    Then if $\querytime$ is polylogarithmic then $\spacecomplexity = \Omega(n^2)$ on a
    realRAM model of computation.
\end{conjecture*}

Note that there are various unconditional lower bounds for Hopcroft's problem.
We refer to the article by Erickson~\cite{Jeff2000} for detailed lower
bounds.
While Erickson~\cite{Jeff2000} states many different lower bounds, none of them gives us the type of lower bound that we need. 
The reason is that his lower bounds do not match our setting
in at least one of the following points.
\begin{itemize}
    \item We do not know in which way the coordinates are given. We need the coordinates to be not so large integers.
    \item Some of Erickson's lower bounds are about preprocessing time and we care more about the space complexity.
    \item Some of Erickson's lower bounds are about different algorithmic problems, like line queries or half-plane emptiness queries. 
    We are only concerned with line-emptiness queries.
    \item Erickson's lower bounds are about semi-groups or
    partition graphs and
    we do not know how to compare those to the real RAM model.
\end{itemize}
% 
% 
% 
% We define the \TwoPointShortestPath problem as the 
% data structure problem with the input being 
% a polygonal domain $P$ on $n$ vertices and we permit queries of 
% two points $s,t$ and we expect 
% to receive the length of the shortest path within $P$ from $s$ to $t$. \sarita{this is now already defined in the introduction right?}
% We denote by $p$ the preprocessing time and by $t$ the query time.

\lowerbound*

% \begin{theorem}[Lower Bounds]
%   \label{thm:LowerBound}
%     Consider an data structure for the \TwoPointShortestPath paroblem with
%     preprocessing time $p$ and query time $t$.
%     Assuming the Integer Hopcroft lower bound, 
%     it holds that the 
%      $p + tn = \Omega(n^2)$.
% \end{theorem}

% \frank{the notation $t$ conflicts with the def of $t$ as one of the
%   query points. Maybe use ``preprocessing time $T$ and query time $Q$?''} \sarita{Both $p$ and $t$ are also dependent on $n$, so I would suggest preprocessing time $T(n)$ and query time $Q(n)$}
%   \till{I am creating a command and then we can easily adopt it.
%   But we should supress most dependencies. It also depends on teh polygon P, but we do not want to mention this for sanity.}

%Note that the theorem implies that any sublinear upper bound on the query time implies a quadratic lower bound on the preprocessing time.

\begin{figure}
    \centering
    \includegraphics{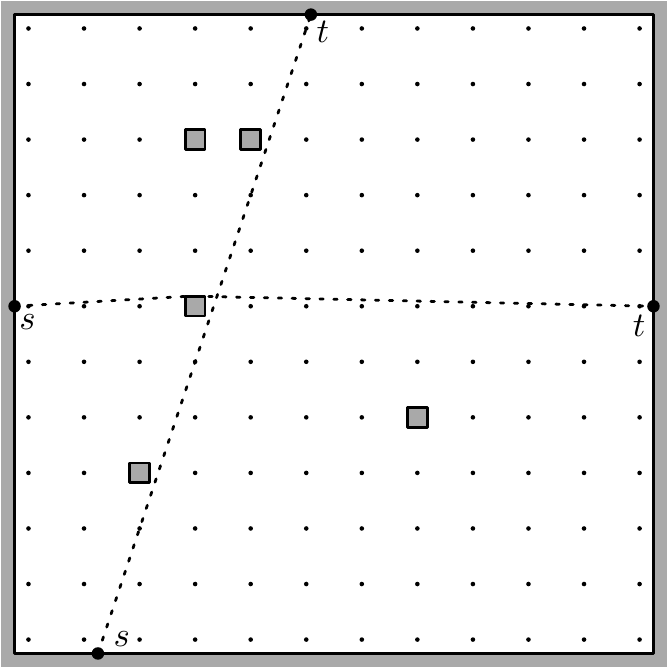}
    \caption{We replace each point by a small square such that the the geodesic distance between $s$ and $t$ becomes larger whenever the line through $s$ and $t$ intersects a point.}
    \label{fig:enter-label}
\end{figure}

To attain this lower bound, we reduce from Hopcroft's problem. The first part of the reduction is a construction of a polygonal domain and queries.
Thereafter, we will show that this will indeed solve Hopcroft's problem.

\subparagraph*{Construction.}
Assume we are given a set $R$ of $n$ points as above, we are now constructing a polygonal domain $P$ such that any line emptiness query on $R$ can be answered using a shortest path query for $P$.
To be precise, $P$ consists of a large boundary square with coordinates $(-L,-L)$, $(L,-L)$, $(-L,L)$, and $(L,L)$. 
Furthermore, we create a small hole for each point $p\in S$. The hole is a square of sidelength $a$ (to be determined later) and center $p$.
For each line $\ell \in \L$, we can find whether or not the line contains a point by a shortest path query between the points $s$ and $t$ where $\ell$ intersects the outer square of $P$.
We compare the shortest distance between $s$ and $t$ to their Euclidean distance.
If the geodesic distance is larger than the Euclidean distance, we say that the line contains a point. 
By performing this query for every line in $\L$, we can conclude that no line contains a point.

\subparagraph*{Correctness.}
We are now ready to show that it solves the online integer version of Hopcroft's  problem.
Consider some line $\ell \in \L$ that contains a point $p \in S$ and the corresponding query points $s,t$.
It is easy to see that the geodesic distance will be different to the Euclidean distance as path from $s$ to $t$ needs to navigate around the square with center $p$.

Now, let's consider the situation that point $p$ is not on line $\ell$. 
Then we know that the distance between $p$ and $\ell$ is at least $L^{-5}$~\cite[Lemma 5]{ArtApprox}.
By setting the side length of all small square to $a = L^{-5}/4$ we can ensure that $\ell$ will not intersect this square.
As we set the side length of every square it follows that $\ell$ will not intersect any square.
And thus the geodesic distance equals the shortest path distance.

\subparagraph*{Conclusion.}
The construction of $P$ allows us to solve the online integer version of Hopcroft's problem with polylogarithmic query time.
Thus, the space complexity of our data structure must be at least $\Omega(n^2)$ according to the Integer Hopcroft Lower Bound conjecture.

\newpage
\bibliography{bibliography.bib}

\end{document}